\documentclass[10pt, final]{IEEEtran}
\usepackage{amsmath,graphicx}

\usepackage{amsthm}
\usepackage{hyperref}
\usepackage{amssymb}
\usepackage[lined,boxed,commentsnumbered, ruled]{algorithm2e}
\usepackage{mathrsfs}
\usepackage{algorithmic}
\usepackage{bm}
\usepackage{pgfplots}
\usepackage{tikz}
\usetikzlibrary{arrows}
\usepackage{microtype} \DisableLigatures[f]{encoding = *, family= *}
\usepackage{booktabs,multirow}
\usepackage{epstopdf}
\usepackage{graphicx}
\usepackage{subcaption}
\usepackage{subeqnarray}
\usepackage[switch]{lineno}

\definecolor{colorhkust}{RGB}{20,43,140}
\definecolor{colortsinghua}{RGB}{116,52,129}
\definecolor{color1}{RGB}{128,0,0}

\newtheorem{proposition}{Proposition}

\newcommand{\K}{\mathcal{K}}
\newcommand{\N}{\mathcal{N}}
\newcommand{\M}{\mathcal{M}}
\newcommand{\A}{\mathcal{A}}
\newcommand{\B}{\mathcal{B}}

\newcommand{\C}{\mathbb{C}}
\newcommand{\D}{{\textrm{DL}}}
\newcommand{\U}{\sf{\textrm{UL}}}
\newcommand{\V}{\sf{\textrm{VDL}}}
\renewcommand{\H}{\sf{H}}
\begin{document}
\title{Reconfigurable Intelligent Surface for Green Edge Inference}

\author{\IEEEauthorblockN{Sheng Hua,~\IEEEmembership{Student Member,~IEEE,} Yong Zhou,~\IEEEmembership{Member,~IEEE,} Kai Yang,~\IEEEmembership{Student Member,~IEEE,}\\ and Yuanming Shi,~\IEEEmembership{Member,~IEEE}}\\
	\thanks{S. Hua, Y. Zhou, and Y. Shi are with the School of Information Science and Technology, ShanghaiTech University, Shanghai 201210, China (E-mail:\{huasheng, zhouyong, shiym\}@shanghaitech.edu.cn).}
	\thanks{
		K. Yang is with the School of Information Science and Technology, ShanghaiTech University, Shanghai 201210, China, also with the Shanghai Institute of Microsystem and Information Technology, Chinese Academy of Sciences, Shanghai 200050, China, and also with the University of Chinese Academy of Sciences, Beijing 100049, China (E-mail: yangkai@shanghaitech.edu.cn).
	}
	\thanks{This paper will be presented in part at the \textit{IEEE Globecom Workshops}, Waikoloa, Hawaii, Dec. 2019. \cite{hua2019globecom}}
	
}

%
% Single address.
% ---------------
%\author{names}{Gao Yin$^{*}$, Kai Yang$^{*}$, Yuanming Shi$^{*}$}%, and Zhi Ding$^{\dag}$
%\address{$^{*}$School of Information Science and Technology,
%ShanghaiTech University, Shanghai, China\\
%%$^{\dag}$Dept. of ECE, University of California, Davis, California 95616, USA\\
%                           E-mail: \{yangkai, shiym\}@shanghaitech.edu.cn, %zding@ucdavis.edu\\
% }

% \name{Kai~Yang\thanks{K.
% Yang is with the School of Information Science and Technology,
% ShanghaiTech University, Shanghai, China (e-mail: yangkai@shanghaitech.edu.cn).}}
% \address{ShanghaiTech University, Shanghai, China}

\maketitle
%\linenumbers

\begin{abstract}
Reconfigurable intelligent surface (RIS) as an emerging cost-effective technology can enhance the spectrum- and energy-efficiency of wireless networks. In this paper, we consider an RIS-aided green edge inference system, where the inference tasks generated from resource-limited mobile devices (MDs) are uploaded to and cooperatively performed at multiple resource-enhanced base stations (BSs). Taking into account both the computation and uplink/downlink transmit power consumption, we formulate an overall network power consumption minimization problem, which calls for the joint design of the set of tasks performed by each BS, transmit and receive beamforming vectors of the BSs, transmit power of the MDs, and uplink/downlink phase-shift matrices at the RIS. Such a problem is a mixed combinatorial optimization problem with nonconvex constraints and is highly intractable. To address the challenge of the combinatorial objective, a group sparse reformulation is proposed by exploiting the group sparsity structure of the beamforming vectors, while a block-structured optimization (BSO) approach is proposed to decouple the optimization variables. Finally, we propose a BSO with mixed $\ell_{1,2}$-norm and difference-of-convex-functions (DC) based three-stage framework to solve the problem, where the mixed $\ell_{1,2}$-norm is adopted to induce the group sparsity of beamforming vectors and DC is adopted to effectively handle the nonconvex rank-one constraint after matrix lifting. Numerical results demonstrate the supreme gain of deploying an RIS and confirm the effectiveness of the proposed algorithm over the baseline algorithms.
\end{abstract}

\begin{IEEEkeywords}
Reconfigurable intelligent surface, joint uplink and downlink, green edge inference, block-structured optimization, mixed $\ell_{1,2}$-norm, difference-of-convex programming
\end{IEEEkeywords}

\section{Introduction}
Benefiting from the availability of big data, recent years have witnessed the rapid prosperity of deep neural network (DNN), which has demonstrated its superiority in a variety of intelligent applications (e.g., computer vision and natural language processing). Meanwhile, with an ever-increasing number of mobile devices (MDs) that will generate 77 exabytes data per month by 2022 \cite{CiscoReport}, the demand of performing inference tasks (e.g., object recognition and machine translation) is anticipated to be ubiquitous, especially in the artificial intelligence (AI)-powered sixth generation (6G) networks \cite{letaief2019roadmap}. Driven by this trend, it is urgent to push traditional cloud-based DNN models to the network edge so as to unleash the potentials of edge data and in turn provide intelligent services \cite{chen2019deeplearning}-\cite{zhang2019iov}. One possible architecture to perform inference tasks is on-device inference, i.e., running DNN models directly on MDs. While the model compression \cite{liu2018demand}, model selection \cite{taylor2018adaptive}, and hardware acceleration \cite{du2015shidiannao} have been proposed as promising techniques to help MDs run DNN models, deploying powerful models with real-time execution requirements is still challenging because of the resource limitations of MDs \cite{chen2019deeplearning}.

%Unfortunately,  In addition, resorting to the cloud servers suffers from excessive response latency due to the long propagation distance \cite{chen2019deeplearning}.  Moreover, are operated in a cloud-centric way, 

%In order to perform intelligent tasks, data gathered by devices must be transmitted to remote servers, which introduces new issues such as latency because of long propagation distance, scalability because of the scarcity access to the cloud servers, and privacy concerns at the risk of sensitive information leakage . On the other hand, the idea of running a DNN model on device is prohibitive due to the limited resources such as storage, computation and battery of mobile devices.
%For instance, the parameter size of AlexNet is $62$ million \cite{AlexNetsize}, which takes up to $250$ megabytes to store them. However the storage capacity of a device is in the order of several megabytes. 

By leveraging edge computing \cite{shi2016edge} and deploying DNN models at the edge base stations (BSs) that have strong computational capacity and large storage resources, edge inference stands out as a promising paradigm to provide intelligent services for MDs \cite{zhou2019edge}.  To accomplish the inference tasks, the MDs upload the task-specific data to the BSs and subsequently the BSs deliver the inference results after finishing the inference process. Tailored for latency-critical applications, the authors proposed device-edge \cite{li2019edge} and edge-cloud \cite{hu2019dynamic} synergy frameworks to partition DNN model parameters based on network dynamics so as to minimize the execution latency. As energy efficiency is a key performance indicator for edge inference systems \cite{zhou2019edge}, the authors in \cite{cvpr_2017_yang_energy}-\cite{louizos2017bayesian} proposed energy-aware approaches to prune DNN models so as to minimize the computation power consumption (i.e., power required for the BSs to perform the inference tasks) while maintaining reasonable inference precision. However, the communication power consumption is not considered in \cite{cvpr_2017_yang_energy}-\cite{louizos2017bayesian}. The authors in \cite{yang2019energy} proposed to minimize the sum of computation power and downlink transmit power consumption (i.e., power required for the BSs to deliver the tasks results to the MDs), while the uplink transmit power (i.e., power required for the MDs to upload data to the BSs) is neglected. However, the uplink traffic load (e.g., raw images for an object recognition task) is usually comparable to the downlink one (e.g., labeled images) in edge inference systems, resulting in high uplink transmit power consumption. Therefore, it is imperative to develop new techniques to reduce both the uplink and downlink transmit power consumption and in turn facilitate an energy-efficient design for edge inference systems.

%To improve the users' quality of service (QoS), we explore the idea of computation replication. Specifically, the MDs are allowed to upload tasks to multiple BSs to create multiple copies of the computed results at different BSs thereby enabling transmission cooperation. Apparently, the better QoS is achieved at the cost of computation inefficiency. As performing inference tasks are usually power-demanding due to the complexity of DNN models, it is essential to be power-aware and improve energy-efficiency. To save power, each BS can selectively perform only a part of inference tasks, not all of them.

Recently, a growing line of works focused on an emerging technology named reconfigurable intelligent surface (RIS \cite{basar2019wireless}, or reconfigurable meta-surface \cite{di2019smart}, hypersurface \cite{liaskos2018new}), which has the potential to significantly reduce the power consumption \cite{wu2018intelligent}
%\cite{wu2018intelligent,fu2019intelligent,han2019intelligent,li2019joint} 
and improve the energy efficiency \cite{huang2019reconfigurable}. An RIS is a low-cost planar array consisting of a large number of passive reflecting elements with reconfigurable phase shifts, each of which can be dynamically tuned via a software controller to reflect the incident signals \cite{nadeem2019asymptotic}-\cite{huang2019holographic}. These elements consume negligible energy due to their passive nature. By adaptively adjusting the phase shifts of reflecting elements, an RIS can combine the constructive signals and suppress the interference, thereby greatly enhancing the performance of wireless systems \cite{di2019smart}-\cite{liaskos2018new}. By jointly optimizing the beamforming vectors and the phase-shift matrix, deploying an RIS has been shown as an effective way to reduce the power consumption in various applications, e.g., downlink unicast \cite{wu2018intelligent} and broadcast \cite{han2019intelligent} settings, non-orthogonal multiple access \cite{fu2019Intelligent}-\cite{li2019joint}, and simultaneous wireless information and power transfer \cite{wu2019joint}. All the previous works only considered the downlink transmit power consumed by the BSs. However, in edge inference systems, the computation power consumption is an indispensible component and should be taken into account to accurately characterize the overall network power consumption. In addition, it is essential to optimize both the uplink and downlink phase-shift matrices of the RIS to assist both the uplink and downlink data transmissions. These two key issues make the approaches proposed in the existing works not applicable to our work.

To guarantee the quality of intelligent services provided for MDs, we explore the idea of computation replication \cite{li2019exploiting}, which allows inference tasks to be performed by multiple BSs to create multiple copies of the inference results at different BSs. These copies enable cooperative downlink transmission among the BSs on delivering inference results. In terms of power consumption, however, cooperative transmission and computation replication conflict with each other. Specifically, cooperative transmission reduces downlink transmit power consumption by exploiting a higher beamforming gain, while computation replication rapidly increases the computation power consumption because of repeatedly running DNN models for multiple times. Therefore, we should strike a balance between the computation and communication power consumption via selecting inference tasks performed by each BS and in turn achieve green edge inference.

In this paper, we consider an RIS-aided green edge inference system with multiple BSs cooperatively performing inference tasks for multiple MDs, taking into account both uplink and downlink transmit power consumption as well as the computation power consumption.  Our objective is to minimize the overall network power consumption subject to prescribed quality-of-service (QoS) requirements, by jointly designing the task selection strategy, transmit/receive beamforming vectors of the BSs, the transmit power of the MDs, and the uplink/downlink phase-shift matrices at the RIS. However, the formulated problem is a mixed combinatorial optimization problem with nonconvex constraints, which is highly intractable.

%In this paper, we focus on the  problem in an RIS-assisted edge inference system, taking into account both the computation and communication power consumption. The computation power consumption is determined by the set of performed tasks, while the communication power consumption consists of both the uplink (i.e., the MDs transmit data to the BSs) and in the downlink (i.e., ) are considered.  However, , for which we propose an effective algorithm design.
\subsection{Contributions}
The main contributions of this paper are summarized as follows.
%(1) We propose an energy-efficient processing and transmission approach for BSs in a multi-cell communication system collaboratively serving multiple MDs to perform intelligent inference tasks. An RIS is deployed at the cell boundary to alleviate harsh signal propagation environments and improve service quality.
\begin{itemize}
	\item We propose a joint design of the task selection strategy, transmit/receive beamforming vectors, transmit power, and uplink/downlink phase-shift matrices for an RIS-aided green edge inference system. To the best of our knowledge, this is the first attempt to unify beamforming vectors, transmit power, and phase shifts design in both the uplink and downlink transmissions into a general framework.
	\item The combinatorial nature of the task selection strategy and the coupled optimization variables stand out as two major challenges. We address the challenge of the combinatorial variable by exploiting the group sparsity structure of the beamforming vectors, and tackle the challenge of the coupled variables by proposing a block-structured optimization (BSO) approach.
	\item With fixed phase shifts, we adopt the weighted mixed $\ell_{1,2}$-norm to induce the group sparsity of beamforming vectors. With fixed beamforming vectors and transmit power, the original problem is transformed to a homogeneous quadratically constrained quadratic programming (QCQP) with a nonconvex rank-one constraint. As the widely adopted semidefinite relaxation (SDR) technique incurs performance degradation when the number of reflecting elements is large, we propose a novel difference-of-convex-functions (DC) representation for this nonconvex constraint, followed by proposing an effective DC algorithm. We then propose a BSO with mixed $\ell_{1,2}$-norm and DC based three-stage framework to solve the original problem.
	\item Through extensive simulations, we show that the deployment of an RIS can significantly reduce the overall network power consumption. Furthermore, the proposed BSO with mixed $\ell_{1,2}$-norm and DC algorithm achieves a significant performance improvement compared to BSO with mixed $\ell_{1,2}$-norm and SDR algorithm, which demonstrates the effectiveness of DC in yielding the rank-one solutions.
\end{itemize}

\subsection{Organization and Notations}
\textit{Organization:} The remainder of this paper is organized as follows. We present the system model and problem formulation in Section \ref{sec: system model}. A BSO approach is developed in Section \ref{sec: alternating approach}. Based on mixed $\ell_{1,2}$-norm and DC approach, we propose a three-stage framework in Section \ref{sec: Framework}. Simulation results are illustrated in Section \ref{sec: simulation}. Finally, Section \ref{sec: conclusion} concludes this paper.

\textit{Notations:} We use boldface lower-case (e.g., $\bm{h}$) and upper-case letters (e.g., $\bm{G}$) to represent vectors and matrices, respectively. The transpose, conjugate transpose, trace operator and diagonal matrix are denoted as $(\cdot)^{{\sf{T}}},(\cdot)^{{\H}},\mathrm{Tr}(\cdot)$ and  $\mathrm{diag}(\cdot)$, repectively. The symbols $|\cdot|$ and $\mathfrak{R}(\cdot)$ denote the modulus and the real component of a complex number. The $n\times n$ identity matrix is denoted as $\bm{I}_n$. The complex normal distribution is denoted as $\mathcal{CN}$. The inner product of two matrices $\bm{X}$ and $\bm{Y}$ is denoted as $\left\langle \bm{X}, \bm{Y} \right\rangle$, which is defined as $\left\langle \bm{X}, \bm{Y} \right\rangle = \mathrm{Tr}(\bm{X}^{\H}\bm{Y})$. The $\ell_2$-norm of a vector is denoted as $\|\cdot\|_2$. The spectral norm and Frobenius norm of a matrix are denoted as $\|\cdot\|$ and $\|\cdot\|_{F}$, respectively. The $i$-th largest singular value of matrix $\bm{X}$ is denoted as $\sigma_{i}(\bm{X})$. We use $\bm{1}_{\{\cdot\}}$ to denote the indicator function which outputs $1$ if the condition $\cdot$ is satisfied, and outputs $0$ otherwise. In the rest of this paper, the superscripts $\textrm{UL}$ and $\textrm{DL}$ refer to uplink and downlink, respectively, and the letters $\textrm{d}$ and $\textrm{r}$ in the subscripts stand for the \textit{direct link} and the \textit{reflected link}, respectively.

\section{System Model and Problem Formulation}\label{sec: system model}
\begin{figure}[!t]
	\centering
	\includegraphics[width=\linewidth]{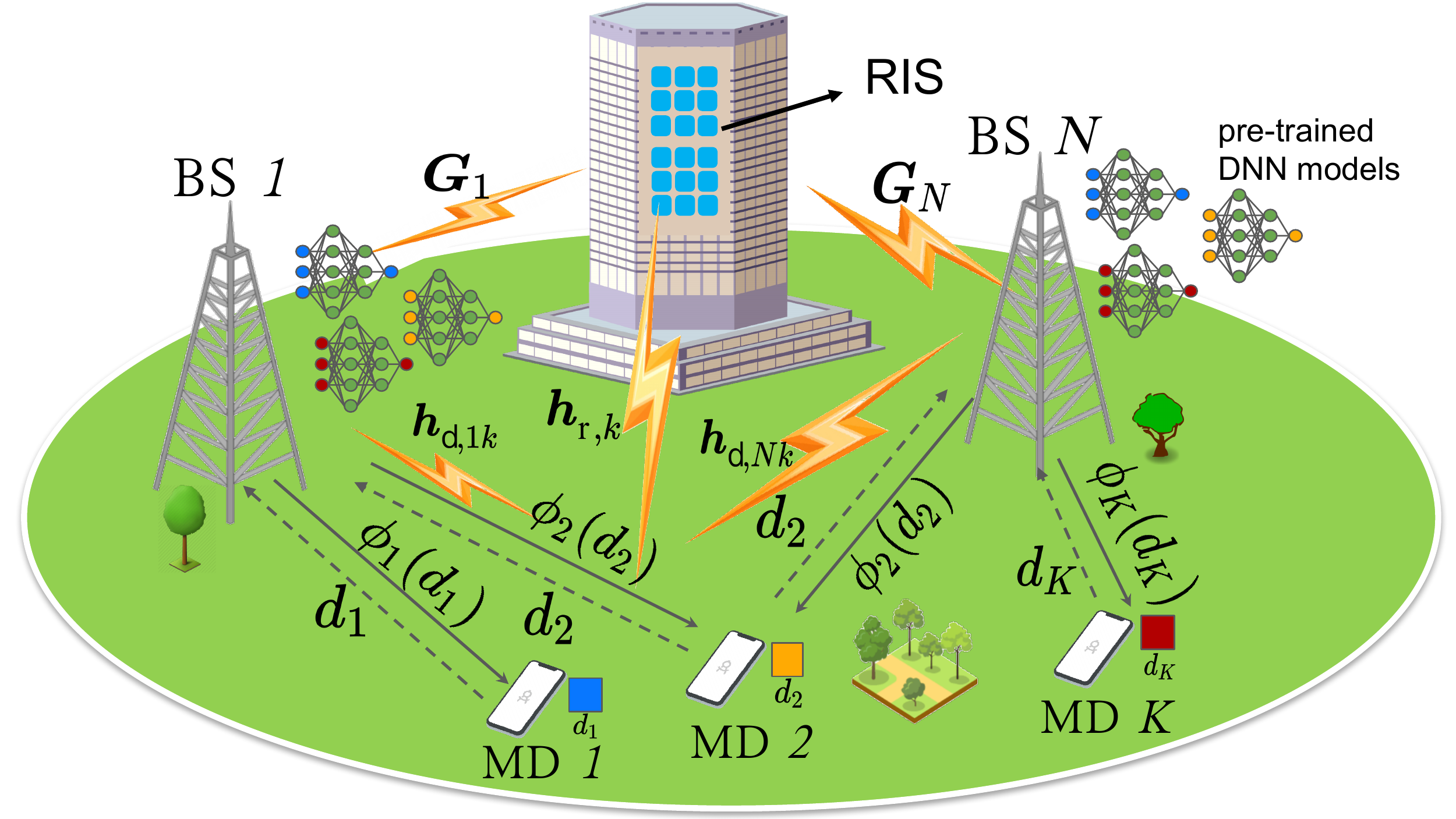}
	\caption{A reconfigurable intelligent surface (RIS)-aided edge inference system with $N$ base stations (BSs) collaboratively serving $K$ mobile devices (MDs). The RIS is deployed on the facade of a building.}
	\label{fig:sys model}                                                                    
\end{figure}
In this section, we describe the system model and the power consumption model for performing inference tasks at the network edge, followed by an overall network power minimization formulation for green edge inference while guaranteeing the quality of intelligent services.

\subsection{System Model}
We consider an RIS-aided edge computing system, where $N$ $L_n$-antenna BSs distributed in a small-cell network collaboratively serve $K$ single-antenna MDs with the assistance of an $M$-element RIS deployed on the facade of a building, as shown in Fig. \ref{fig:sys model}. Let $\mathcal{N}=\{1,\dots,N\}$,  $\mathcal{K}=\{1,\dots,K\}$, and $\mathcal{M}=\{1,\dots,M\}$ denote the index sets of BSs, MDs, and reflecting elements, respectively.
The BSs are resource-enhanced with strong computation and storage capabilities \cite{mao2017survey}.
Each MD has an inference task (e.g., image recognition) to be processed by a task-specific DNN model (e.g., AlexNet \cite{krizhevsky2012imagenet}). Specifically, the DNN model denoted as $\phi_k$ takes MD $k$'s local data $d_k$ (e.g., raw images) as input and generates the inference result $\phi_k(d_k)$ (e.g., labeled images) as output. As it is infeasible to run DNN models on resource-constrained MDs, we in this paper propose to perform inference tasks for the MDs at the BSs.  We assume that all the BSs have downloaded the pre-trained DNN models from cloud servers in advance \cite{yang2019energy}.

%traditional cloud-based solutions to perform DNN inference in the powerful cloud datacenter are prohibitive due to the unbearable transmission latency \cite{li2019edge}. 

%Instead, in edge computing, the resource-enabled BSs are located in user proximity and serve as an extension of cloud datacenter at the network edge therefore they can provide cloud services for the MDs \cite{shi2016edge, mao2017survey}. 
%We assume that all BSs have downloaded and cached the pre-trained models tailored for the MDs' tasks in advance, and the training phase to obtain these models has taken place in the cloud datacenter.

The overall process of accomplishing the inference tasks in the edge computing system is composed of the following three phases.
\begin{itemize}
	\item \textit{Uplink Transmission:} The MDs upload the collected input data $\{d_k,\,k\in\K\}$ to the BSs.
	\item \textit{Inference Computation:} The BSs feed data (e.g., $d_k$) into a specific pre-trained DNN model (e.g., $\phi_k$) according to the task type and then obtain the inference results (e.g., $\phi_k(d_k)$).
	\item \textit{Downlink Transmission: }The BSs deliver the inference results $\{\phi_k(d_k),\,k\in\K\}$ to the corresponding MDs.
\end{itemize}
By exploiting the broadcast nature of wireless channels, each MD's data can be successfully received by multiple BSs in the uplink, which enables the computation replication and creates multiple copies of the inference results at different BSs \cite{li2019exploiting}. In the downlink, the BSs performing the same inference task cooperatively transmit the inference results to the corresponding MD \cite{gesbert2010multi}. To enable transmission cooperation among the BSs, it is assumed that the global channel state information (CSI) is available at the BSs.
Let $\A_n\subseteq\K$ denote the set of MD indices whose inference tasks are selectively performed by BS $n$. For notational ease, we denote $\A=\left(\A_1,\dots,\A_N\right)$ as the task selection strategy.

\subsubsection{Uplink Transmission}
%In addition, the channels are assumed as quasi-static fading channels.
Let $s_k^{\U}\in\C$ denote the representative information symbol of input data $d_k$, and $p_k^{\U}\in\mathbb{R}$ denote the transmit power of MD $k$. Without loss of generality, $\{s_k^{\U}, k\in\K\}$ are assumed to have zero mean and unit power.
The signal received at BS $n$ can be expressed as 
\begin{eqnarray}\label{eq: recv signal}
\bm{y}_n^{\U} = \sum_{k\in\K}\bm{g}_{nk}^{\U} \sqrt{p_k^{\U}}s_k^{\U} + \bm{z}_n^{\U}, 
\end{eqnarray}
where $\bm{g}_{nk}^{\U}\in\C^{ L_{n}\times 1}$ is the equivalent baseband channel response from MD $k$ to BS $n$ and $\bm{z}_n^{\U}\sim\mathcal{CN}\left(\bm{0},\sigma_n^2\bm{I}_{L_n}\right)$ is the additive white Gaussian noise (AWGN) at BS $n$ with $\sigma_n^2$ being the noise power. With the deployment of an RIS, the equivalent baseband channel from MD $k$ to BS $n$ consists of both the direct link and the reflected link, where the reflected link is a concatenation of the MD-RIS link, the phase shifts at the RIS, and the RIS-BS link \cite{wu2018intelligent}, \cite{han2019intelligent}-\cite{wu2019joint}. Therefore, $\bm{g}_{nk}^{\U}$ can be modeled as
\begin{eqnarray} \label{eq: uplink channel}
\bm{g}_{nk}^{\U}=\underbrace{\bm{h}_{\textrm{d},nk}^{\U}}_{\textrm{direct link}}+~ \underbrace{\left(\bm{G}^{\U}_{n}\right)^{\H}\left(\bm{\Theta}^{\U}\right)^{\H}\bm{h}_{\textrm{r},k}^{\U}}_{\textrm{reflected link}},
\end{eqnarray}
where $\bm{h}_{\textrm{d},nk}^{\U}\in\mathbb{C}^{L_n \times 1}$, $\bm{h}_{\textrm{r},k}^{\U}\in \C^{M\times 1}$, and $\bm{G}_n^{\U}\in \C^{M\times L_n}$ denote the channel responses from MD $k$ to  BS $n$, from MD $k$ to the RIS, and from the RIS to BS $n$, respectively. In addition, $\bm{\Theta}^{\U}=\beta\textrm{diag}\left( \theta_1^{\U},\dots, \theta_M^{\U}\right)\in\C^{M\times M}$ denotes the diagonal phase-shift matrix for uplink transmission, where $\beta\in\left[0,1\right]$ is the amplitude reflection coefficient and $\theta_m^{\U}=e^{j\varphi_m^{\U}}$ with  $\varphi_m^{\U}\in\left[0,2\pi\right)$ being the uplink phase shift of the $m$-th reflecting element of the RIS. The reflected link only accounts for one-time reflection, because the power of signals reflected by two or more times is negligible due to the high path loss \cite{wu2018intelligent}, \cite{han2019intelligent}-\cite{wu2019joint}.

We consider the linear beamforming strategy, and denote the receive beamforming vector of BS $n$ to decode $s_k^{\U}$ as $\bm{v}_{nk}^{\U} \in \C^{L_n \times 1}$. BS $n$ only decodes MD $k$'s transmitted symbol $s_k^{\U}$ if $k\in\A_n$. The estimated symbol at BS $n$ for MD $k\in\A_n$, denoted by $\hat{s}_{nk}^{\U}\in\mathbb{C}$, is given by
\begin{eqnarray}
%&\left(\bm{v}_k^{\U}\right)^{{\H}}\bm{y}^{\U} \nonumber \\
\hspace{-2em}&&\hat{s}_{nk}^{\U}=\left(\bm{v}_{nk}^{\U}\right)^{{\H}}\bm{y}_n^{\U}
= \left(\bm{v}_{nk}^{\U}\right)^{{\H}}\bm{g}_{nk}^{\U}\sqrt{p_k^{\U}}s_k^{\U} \nonumber \\
\hspace{-2em}&&\quad\quad~~+ \left(\bm{v}_{nk}^{\U}\right)^{{\H}}\sum_{l\neq k}\bm{g}_{nl}^{\U} \sqrt{p_l^{\U}}s_l^{\U} +  \left(\bm{v}_{nk}^{\U}\right)^{{\H}}\bm{z}_{n}^{\U}.
\end{eqnarray}

Therefore, the uplink signal-to-interference-plus-noise ratio (SINR) observed at BS $n$ for MD $k\in\A_n$ is 
\begin{eqnarray}\label{eq: uplink sinr}
\textrm{SINR}_{nk}^{\U} = \frac{p_k^{\U}\left|\left(\bm{v}_{nk}^{\U}\right)^{{\H}}\bm{g}_{nk}^{\U}\right|^2}{\sum_{l\neq k}p_l^{\U}\left|\left(\bm{v}_{nk}^{\U}\right)^{{\H}}\bm{g}_{nl}^{\U}\right|^2+\sigma_n^2\left\|\bm{v}_{nk}^{\U}\right\|_2^2}.
\end{eqnarray}

\subsubsection{Downlink Transmission}
After performing the inference tasks, the BSs cooperatively transmit the inference results $\{\phi_k(d_k), k\in\K\}$ to the corresponding MDs through downlink wireless channels.
%Similarly, the equivalent downlink channel from BS $n$ to MD $k$ can be expressed as follows
%\begin{eqnarray}\label{eq: downlink channel}
%\left(\bm{g}_{kn}^{\D}\right)^{\H}=\left(\bm{h}_{\textrm{d},kn}^{\D}\right)^{\H}+ \left(\bm{h}_{\textrm{r},k}^{\D}\right)^{\H}\bm{\Theta}^{\D}\bm{G}^{\D}_{n} \in\C^{1\times L_{n}}.
%\end{eqnarray}
Let $s_k^{\D}\in\C$ denote the representative symbol of $\phi_k(d_k)$ intended for MD $k$ and $\bm{v}_{nk}^{\D}$ denote the downlink beamforming vector from BS $n$ to MD $k$. Without loss of generality, $\{s_k^{\D},~k\in\K\}$ are assumed to have zero mean and unit power. The signal transmitted by BS $n$, denoted as $\bm{x}_n^{\D}\in \C^{L_n\times 1}$, is a summation of beamformed symbols for MD $k\in\A_n$, i.e.,
\begin{eqnarray}
\bm{x}_n^{\D} = \sum_{k\in\A_n}\bm{v}_{nk}^{\D} s_k^{\D},~\forall\, n\in\N.
\nonumber
\end{eqnarray}
%The signal from BS $n$ is transmitted to MD $k$ if $k\in\A_n$. 
The signal received by MD $k$ can be expressed as
\begin{eqnarray}\label{eq: downlink signal}
&&\hspace{-1.8em}y_k^{\D} =\sum_{n\in\N}\left(\bm{g}_{nk}^{\D}\right)^{\H}\bm{x}_n^{\D}+z_k^{\D}\nonumber \\
%=\sum_{n=1}^{N}\bm{1}_{\{k\in\A_n\}}\left(\bm{g}_{nk}^{\D}\right)^{\H}\bm{v}_{nk}^{\D} s_k^{\D}\nonumber \\ &&\quad\quad\quad+\sum_{n=1}^{N}\sum_{l\in\A_n,\,l\neq k}\left(\bm{g}_{nk}^{\D}\right)^{\H}\bm{v}_{nl}^{\D} s_l^{\D}+ z_k^{\D} \nonumber \\ &&
&&=\sum_{n\in\N}\left(\bm{g}_{nk}^{\D}\right)^{\H}\!\left(\bm{1}_{\{k\in\A_n\}}\bm{v}_{nk}^{\D} s_k^{\D}\!+\sum_{\underset{l\in\A_n}{l\neq k,}}\bm{v}_{nl}^{\D} s_l^{\D}\right)+z_k^{\D}
\nonumber \\ && =\sum_{n\in\N}\bm{1}_{\{k\in\A_n\}}\left(\bm{g}_{nk}^{\D}\right)^{\H}\bm{v}_{nk}^{\D} s_k^{\D} \nonumber \\
&&\quad+ \sum_{l\neq k}\sum_{n\in\N}\bm{1}_{\{l\in\A_n\}}\left(\bm{g}_{nk}^{\D}\right)^{\H}\bm{v}_{nl}^{\D}s_l^{\D}+z_k^{\D},
\end{eqnarray}
where $z_k^{\D}\in\mathbb{C}$ is the AWGN at MD $k$ with zero mean and power $\sigma_{k}^2$, and $\left(\bm{g}_{nk}^{\D}\right)^{\H}\in\mathbb{C}^{1\times L_n}$ is the equivalent downlink channel response from BS $n$ to MD $k$. Similar to the uplink counterpart, $\left(\bm{g}_{nk}^{\D}\right)^{\H}$ can be modeled as
\begin{eqnarray}\label{eq: downlink channel}
\left(\bm{g}_{nk}^{\D}\right)^{\H} = \underbrace{\left(\bm{h}_{\textrm{d},nk}^{\D}\right)^{\H}}_{\textrm{direct link}}+~ \underbrace{\left(\bm{h}_{\textrm{r},k}^{\D}\right)^{\H}\bm{\Theta}^{\D}\bm{G}^{\D}_{n}}_{\textrm{reflected link}},
\end{eqnarray}
where $\left(\bm{h}_{\textrm{d},nk}^{\D}\right)^{\H}\in\mathbb{C}^{1\times L_n }, \left(\bm{h}_{\textrm{r},k}^{\D}\right)^{\H}\in\mathbb{C}^{1 \times M}, ~\textrm{and}~ \bm{G}^{\D}_{n}\in\mathbb{C}^{M\times L_n}$ denote the channel responses from BS $n$ to MD $k$, from the RIS to MD $k$, and from BS $n$ to the RIS, respectively, and $\bm{\Theta}^{\D}=\beta \textrm{diag}\left(\theta_1^{\D},\dots, \theta_M^{\D} \right)\in\C^{M\times M}$ is the downlink phase-shift matrix with diagonal entries $\theta_m^{\D}=e^{j\varphi_m^{\D}}$ and $\varphi_m^{\D}\in\left[0,2\pi\right)$.

%where $z_k^{\D}$ denotes the AWGN with zero mean and variance $\sigma_k^2$.
%For ease of notation, we define $L=\sum_{n=1}^{N}L_n$. Let $\bm{g}_{k}^{\D}=\left[\left(\bm{g}_{1k}^{\D}\right)^{{\sf{T}}},\dots,\left(\bm{g}_{Nk}^{\D}\right)^{{\sf{T}}}\right]^{{\sf{T}}} \in\C^{L\times 1}$ and $\bm{v}_k^{\D} = [\bm{v}_{1k}^{\D}, \dots, \bm{v}_{Nk}^{\D}] \in\C^{L\times 1}$ denote the concatenated downlink channel response vector and the beamforming vector with respect to MD $k$, respectively. Hence, \eqref{eq: downlink signal} can be equivalently written as
%\begin{eqnarray}
%y_k^{\D}=\left(\bm{g}_{k}^{\D}\right)^{\H}\bm{v}_k^{\D}s_k^{\D}+ \sum_{l\neq k}\left(\bm{g}_{k}^{\D}\right)^{\H}\bm{v}_l^{\D}s_l^{\D} + z_k^{\D},~\forall\,k.
%\end{eqnarray}

%\begin{eqnarray}
%\mathrm{where}~~ \left(\bm{g}_{k}^{\D}\right)^{\H}&=&\left[\left(\bm{h}_{k1}^{\D}\right)^{\H},\dots,\left(\bm{h}_{kN}^{\D}\right)^{\H}\right] \in\C^{1\times L}, \nonumber \\
%\bm{v}_l^{\D} &=& [\bm{v}_{1l}^{\D}, \dots, \bm{v}_{Nl}^{\D}] \in\C^{L\times 1}. \nonumber
%\end{eqnarray}
Based on \eqref{eq: downlink signal}, the SINR observed by MD $k\in\K$ in the downlink is given by
\begin{eqnarray}\label{eq:sinr}
\textrm{SINR}_k^{\D}\!\left(\A\right) =\! \frac{\left|\sum_{n\in\N}\bm{1}_{\{k\in\A_n\}}\!\left(\bm{g}_{nk}^{\D}\right)^{\H}\bm{v}_{nk}^{\D}\right|^2}{\sum_{l\neq k}\left|\sum_{n\in\N}\bm{1}_{\{l\in\A_n\}}\!\left(\bm{g}_{nk}^{\D}\right)^{\H}\bm{v}_{nl}^{\D}\right|^2 \!+\! \sigma_{k}^2}.
\end{eqnarray}
%\begin{equation}\label{eq:sinr}
%\textrm{SINR}_k^{\D} = \frac{\left|\left(\bm{g}_{k}^{\D}\right)^{\H}\bm{v}_k^{\D}\right|^2}{\sum_{l\neq k}\left|\left(\bm{g}_{k}^{\D}\right)^{\H}\bm{v}_l^{\D}\right|^2 + \sigma_{k}^2}, ~\forall\,k\in\K.
%\end{equation}

\subsection{Power Consumption Model}
As running DNN models often incurs high energy consumption due to their high computational complexity \cite{yang2017method, xu2018scaling}, and energy-efficiency is one of the key performance indicators for green communications \cite{mahapatra2015energy}, in this subsection, we present the power consumption model in the proposed edge inference system, taking into consideration both the computation power for inference and the communication power for uplink and downlink transmissions.

\subsubsection{Computation Power Consumption}  We denote the power consumption of performing MD $k$'s inference task at BS $n$ as $P_{nk}^{\sf{c}}$. As a result, the total computation power consumption at all BSs is given by 
\begin{eqnarray}
P_{\textrm{comp}}(\A) = \sum_{n\in\N}\sum_{k\in\A_n}P_{nk}^{\sf{c}}.
\end{eqnarray}

It is worth noting that the majority of the computation power is consumed for running DNN models, and it can be estimated by using the energy estimation methodology proposed in \cite{yang2017method}. This methodology provides a layer-wise energy breakdown for arbitrary neural networks. In particular, the DNN configurations (e.g., number of filters, number of input feature maps) are taken as inputs and the normalized layer-wise consumptions energy of the neural network (i.e., the energy consumption per multiply-and-accumulation (MAC) operation) are generated as outputs \cite{EstimationWebsite}. The computation time can be calculated via dividing the number of MAC operations by the average throughput of a CPU chip. Therefore, the power consumption for performing an inference task equals to the total energy consumption divided by the corresponding computation time.
%The energy consumption of running a DNN consists of the computation part and the data movement part. Specifically, the energy consumption of the computation part is incurred by performing multiplication-and-accumulation (MAC) operations in the arithmetic logic unit, while the energy consumption of the data movement part is due to the access of data in the memory hierarchy. The computation energy is calculated by multiplying the total number of MAC operations and the energy consumption of processing one MAC operation, while the data movement energy is calculated by counting the number of data access at each level of the memory hierarchy and weighing it with the energy consumption of accessing the memory at the corresponding level. 
%Based on this methodology, a lightweight energy estimation tool is developed in \cite{EstimationWebsite}.

For example, the energy consumption of running a widely used DNN model AlexNet to process one image on a well-designed energy-efficient Eyeriss chip can be approximated by that of performing $4\times10^{9}$ MAC operations, and the typical value of the power consumption for processing an image using AlexNet on a Eyeriss Chip at core supply voltage $1.2\,$V is $0.45\,$W \cite{chen2016eyeriss}. As the typical value of computation power $P_{nk}^{\sf{c}}$ (e.g., $0.45\,$W) is comparable to the BSs' transmit power (e.g., $1\,$W \cite{hwang2013holistic}), it is necessary to take into account both the computation and transmit power to facilitate the energy-efficient design.

\subsubsection{Communication Power Consumption} The communication power consumption consists of the power consumed by the MDs in the uplink transmission and by the BSs in the downlink transmission. According to \eqref{eq: uplink channel}-\eqref{eq:sinr}, the total uplink transmit power consumption is $\sum_{k\in\K}p_k^{\U}$, while the downlink transmit power consumption of BS $n$ is given by
\begin{eqnarray}
\mathbf{E}\left[\sum_{k\in\A_n} \left\|\bm{v}_{nk}^{\D}s_k^{\D} \right\|_2^2\right] = \sum_{k\in\A_n} \left\|\bm{v}_{nk}^{\D} \right\|_2^2, \nonumber
\end{eqnarray}
where $\mathbf{E}\left[\cdot\right]$ denotes the expectation. Therefore, the total communication power consumption for both uplink and downlink transmissions is given by 
%It is worth mentioning that the BSs coordinate with each other by sharing merely the CSI, which requires a modest amount of backhaul communication thereby this part of communication cost is negligible \cite{gesbert2010multi}.
\begin{eqnarray}
P_{\textrm{comm}}(\A,\{p_k^{\U}\},\{\bm{v}_{nk}^{\D}\}) = \sum_{k\in\K}p_k^{\U} +  \sum_{n\in\N}\sum_{k\in\A_n}\frac{1}{\eta_n} \left\|\bm{v}_{nk}^{\D} \right\|_2^2, \nonumber
\end{eqnarray}
where $\eta_n$ is the drain efficiency coefficient of the radio frequency power amplifier of BS $n$.

The overall network power consumption, consisting of both the computation power and communication power, can be expressed as
\begin{eqnarray}\label{eq:total power}
&&\hspace{-2.0em}P_{\textrm{total}}(\A,\{p_k^{\U}\},\{\bm{v}_{nk}^{\D}\}) = P_{\textrm{comm}}(\A,\{p_k^{\U}\},\{\bm{v}_{nk}^{\D}\}) \!+\! P_{\textrm{comp}}(\A) \nonumber \\
&&\hspace{-1.8em}= \sum_{k\in\K}p_k^{\U} +  \sum_{n\in\N} \sum_{k\in\A_n}\frac{1}{\eta_n} \left\|\bm{v}_{nk}^{\D} \right\|_2^2 +\sum_{n\in\N}\sum_{k\in\A_n}P_{nk}^{\sf{c}}.
\end{eqnarray}
%where $\eta_n$ is the power amplifier efficiency of BS $n$, and $\lambda$ is the parameter that controls the tradeoff between the power consumptions of the MDs and BSs. Intuitively, a large $\lambda$ is encouraged if we care more about saving MDs' power consumption.

\subsection{Problem Formulation and Analysis}
In the proposed edge inference system, there exists a fundamental tradeoff between communication and computation power consumption. Specifically, with computation replication, more BSs performing the same task can reduce the downlink transmission power consumption by exploiting a higher cooperative beamforming gain, at the cost of increasing the computation power consumption. Therefore, we propose to achieve green edge inference by minimizing the overall network power consumption via striking a balance between the communication and computation power consumption.

Let $\{\gamma_k^{\U},k\in\K\}$ and $\{\gamma_k^{\D},k\in\K\}$ denote the SINR thresholds required to successfully receive the input data and inference results in the uplink and downlink, respectively. The overall network power consumption minimization problem is formulated as 
%\begin{subequations}
%\begin{eqnarray}
%&\underset{\A,\{\bm{v}_k^{\D}\},\bm{\Theta}^{\D},  \{\bm{v}_k^{\U}\}, \{p_k^{\U}\}, , \bm{\Theta}^{\U}}{\mathrm{minimize}} &P_{\textrm{total}} \nonumber\\
%&\mathrm{subject~to} &\textrm{SINR}_k^{\D} \geq \gamma_k^{\D}, ~~\forall k, \\
%\end{eqnarray}
%\end{subequations}
\begin{subequations}
	\begin{eqnarray}\label{prob:original problem}
	%&\mathscr{P}_\textrm{original}:\quad \quad \quad \quad \quad  \nonumber \\
	\hspace{-2.1em}&\mathscr{P}_\textrm{original}:\underset{\{\bm{v}_{nk}^{\U}\}, \{p_k^{\U}\}, \bm{\Theta}^{\U}}{\underset{\A,\{\bm{v}_{nk}^{\D}\},\bm{\Theta}^{\D},}{\mathrm{minimize}}} &P_{\textrm{total}}(\A,\{p_k^{\U}\},\{\bm{v}_{nk}^{\D}\}) \nonumber\\
	\hspace{-5em}&\quad\quad\quad\quad\mathrm{subject~to} &\hspace{-1.2em}\textrm{SINR}_k^{\D}\left(\A\right) \geq \gamma_k^{\D}, ~~\forall\,k\in\K, \label{cons:sinr_downlink_original}\\
	\hspace{-5em}& &\hspace{-2.5em}\textrm{SINR}_{nk}^{\U} \geq \gamma_k^{\U}, \!\forall\,k\in\A_n, n\in\N,~ \label{cons:sinr_uplink_original}\\
	\hspace{-5em}& \quad &\hspace{-2.5em}\sum_{k\in\A_n} \left\|\bm{v}_{nk}^{\D} \right\|_2^2 \leq P_{n,\textrm{max}}^{\D},~\forall\,n\in\N,\, \\
	\hspace{-5em}& \quad \quad \quad \quad \quad &p_k^{\U} \leq P_{k,\textrm{max}}^{\U}, ~~\forall\,k\in\K, \\
	\hspace{-5em}& \quad \quad \quad \quad \quad & |\theta_m^{\D}|=1, ~~\forall\,m\in\M, \label{cons:phase shift_downlink}\\
	\hspace{-5em}& \quad \quad \quad \quad \quad & |\theta_m^{\U}|=1, ~~\forall\,m\in\M,  \label{cons:phase shift_uplink}
	\end{eqnarray}
\end{subequations}
where $P_{k,\textrm{max}}^{\U}$ denotes the maximum transmit power of MD $k$ in the uplink and $P_{n,\textrm{max}}^{\D}$ denotes the maximum transmit power of BS $n$ in the downlink.

In order to solve problem $\mathscr{P}_\textrm{original}$, we are confronted with several main challenges. Problem $\mathscr{P}_{\textrm{original}}$ is a mixed combinatorial optimization problem with nonconvex constraints, which is highly intractable. Besides the troublesome combinatorial variable $\A$, the coupled continuous variables phases shifts and beamforming vectors in constraints \eqref{cons:sinr_downlink_original}-\eqref{cons:sinr_uplink_original} pose a unique challenge. Moreover, the unit-modulus constraints \eqref{cons:phase shift_downlink}-\eqref{cons:phase shift_uplink} imposed by the phase-shift of each RIS element are nonconvex. In the following, we shall propose a group sparse reformulation for the overall network power minimization problem, so as to get rid of the combinatorial variable $\A$ in $\mathscr{P}_\textrm{original}$ and thereby facilitating efficient algorithm design.

%In this section, we propose the joint design of task selection, transmit and receive beamforming vectors at the BSs, transmit power of the MDs, and phase-shift matrices at the RIS to minimize the overall network power consumption in the edge inference system.

\subsection{Group Sparse Reformulation}
By exploiting the intrinsic connection between the task selection (i.e., $\A$) and the group sparsity structure of beamforming vector $\bm{v}_{nk} = \left[\left(\bm{v}_{nk}^{\U}\right)^{{\sf{T}}}, \left(\bm{v}_{nk}^{\D}\right)^{{\sf{T}}}\right]^{{\sf{T}}}$, the combinatorial variable $\A$ can be eliminated. Specifically, if $k\notin\A_n$, BS $n$ does not decode MD $k$'s data in the uplink (i.e., $\bm{v}_{nk}^{\U}=\bm{0}$) and subsequently cannot transmit inference results to MD $k$ in the downlink (i.e., $\bm{v}_{nk}^{\D}=\bm{0}$). In other words, all coefficients in the beamforming group $\bm{v}_{nk}$ are zero simultaneously if $k\notin\A_n$. As multiple tasks may not be performed by a certain BS, the aggregated beamforming vector $\bm{v}\in\C^{K\sum_{n=1}^N L_n}$ defined as
$
\bm{v} = \left[\bm{v}_{11}^{{\sf{T}}},\dots,\bm{v}_{1K}^{{\sf{T}}},\dots, \bm{v}_{N1}^{{\sf{T}}},\dots,\bm{v}_{NK}^{{\sf{T}}}\right]^{{\sf{T}}}
$
is expected to have the group sparsity structure with only a few non-zero blocks.

%Based on these observations, we find that $(n,k)\in\A$ is equivalent to $\bm{1}\left(\bm{v}_{nk}\right)=1$, where $\bm{1}\left(\bm{v}_{nk}\right)$ is the indicator function defined as
%\begin{eqnarray}\label{func: indicator}
%\bm{1}\left(\bm{v}_{nk}\right) = \left\{\begin{array}{ll} {0,} & \mathrm{if}~ \bm{v}_{nk}^{\U}=\bm{v}_{nk}^{\D}=\bm{0} \\ {1,} & \mathrm{otherwise} \end{array} .\right.
%\end{eqnarray}
The above discussions indicate that the task selection strategy $\A$ can be determined by the group sparsity pattern of the beamforming groups, i.e., $\A_n=\{k|\bm{v}_{nk}\neq\bm{0},~k\in\K\},~\forall\,n\in\N$. Therefore, the overall network power consumption \eqref{eq:total power} can be equivalently rewritten as 
\begin{eqnarray}\label{eq: power}
&&\hspace{-3em}P_{\textrm{total}}(\{p_k^{\U}\}, \{\bm{v}_{nk}^{\D}\}) = \sum_{k=1}^{K}p_k^{\U} +  \sum_{n=1}^{N}\sum_{k=1}^{K}\frac{1}{\eta_n} \left\|\bm{v}_{nk}^{\D} \right\|_2^2  \nonumber \\
&&\hspace{5em}\quad~+\sum_{n=1}^{N}\sum_{k=1}^{K}\bm{1}_{\{\bm{v}_{nk}=\bm{0}\}}P_{nk}^{\sf{c}}.
\end{eqnarray}
In addition, the $\textrm{SINR}_{nk}^{\U}$ expression in \eqref{eq: uplink sinr} can be rewritten as 
\begin{eqnarray}
\textrm{SINR}_{nk}^{\U} = \frac{p_k^{\U}\left|\left(\bm{v}_{nk}^{\U}\right)^{{\H}}\bm{g}_{nk}^{\U}\right|^2}{\underset{l\neq k}{\sum}p_l^{\U}\left|\left(\bm{v}_{nk}^{\U}\right)^{{\H}}\bm{g}_{nl}^{\U}\right|^2\!+\!\sigma_n^2\left\|\bm{v}_{nk}^{\U}\right\|_2^2},~\forall\,n, \forall\,k\in\K,~\nonumber
\end{eqnarray}
where $\bm{v}_{nk}^{\U}=\bm{0}$ if $k\notin\A_n$, and $\textrm{SINR}_k^{\D}$ expression in \eqref{eq:sinr} can be simplified as
\begin{eqnarray}
\textrm{SINR}_k^{\D} = \frac{\left|\sum_{n=1}^{N}\left(\bm{g}_{nk}^{\D}\right)^{\H}\bm{v}_{nk}^{\D}\right|^2}{\sum_{l\neq k}\left|\sum_{n=1}^{N}\left(\bm{g}_{nk}^{\D}\right)^{\H}\bm{v}_{nl}^{\D}\right|^2 + \sigma_{k}^2}, ~~\forall\,k, \label{sinr_downlink}
%\\
%= \frac{\left|\left(\bm{g}_{k}^{\D}\right)^{\H}\bm{v}_{k}^{\D}\right|^2}{\sum_{l\neq k}\left|\left(\bm{g}_{k}^{\D}\right)^{\H}\bm{v}_{l}^{\D}\right|^2 + \sigma_{k}^2}, \nonumber
\end{eqnarray}
where $\bm{v}_{nk}^{\D}=\bm{0}$ if $k\notin\A_n$.
%Hence, we can equivalently rewrite \eqref{eq:total power} as 
%\begin{eqnarray}\label{eq:p_total}
%\hspace{-0.5em}P_{\textrm{total}} = \sum_{k=1}^{K}p_k^{\U} +  \sum_{n=1}^{N} \frac{1}{\eta_n}\sum_{k=1}^{K} \left\|\bm{v}_{nk}^{\D} \right\|_2^2 + \sum_{n=1}^{N}\sum_{k=1}^{K}\bm{1}_{\{\bm{v}_{nk}=\bm{0}\}}P_{nk}^{\sf{c}},\nonumber
%\end{eqnarray}
Problem $\mathscr{P}_{\textrm{original}}$ is then equivalent to the following group sparse beamforming problem
\begin{subequations}
	\begin{eqnarray}
	\hspace{-2.5em}&\mathscr{P}:\hspace{0em}\underset{\{\bm{v}_{nk}^{\U}\}, \{p_k^{\U}\}, \bm{\Theta}^{\U}}{\underset{\{\bm{v}_{nk}^{\D}\},\bm{\Theta}^{\D},}{\mathrm{minimize}}} &\hspace{-0.8em}P_{\textrm{total}}(\{p_k^{\U}\}, \{\bm{v}_{nk}^{\D}\}) \nonumber\\
	\hspace{-2.5em}&\quad\quad\mathrm{subject~to} &\hspace{-0.8em}\textrm{SINR}_k^{\D} \geq \gamma_k^{\D}, ~~\forall\,k\in\K,  \label{cons:sinr_downlink}\\
	\hspace{-3em}& &\hspace{-0.8em}\textrm{SINR}_{nk}^{\U} \geq \gamma_k^{\U}, ~~\forall\,n\in\N,~k\in\K,  \label{cons:sinr_uplink}\\
	\hspace{-3em}& &\hspace{-0.8em}\sum_{k=1}^{K} \left\|\bm{v}_{nk}^{\D} \right\|_2^2 \leq P_{n,\textrm{max}}^{\D},~\forall\,n\in\N, \label{cons:power_downlink}\\
	\hspace{-3em}& \quad \quad \quad \quad \quad &\hspace{-0.8em}p_k^{\U} \leq P_{k,\textrm{max}}^{\U}, ~~\forall\,k\in\K, \label{cons:power_uplink}\\
	\hspace{-3em}& \quad \quad \quad \quad \quad & \hspace{-0.8em}|\theta_m^{\D}|=1, ~~\forall\,m\in\M, \label{cons:phase_downlink} \\
	\hspace{-3em}& \quad \quad \quad \quad \quad &\hspace{-0.8em} |\theta_m^{\U}|=1, ~~\forall\,m\in\M,  \label{cons:phase_uplink}
	\end{eqnarray}
\end{subequations}
where we aim to induce the group sparsity of beamforming vector $\bm{v}$. The equivalence means that if $\left(\{{\bm{v}}_{nk}^{\U}\}, \{{p}_k^{\U}\},\{{\bm{v}}_{nk}^{\D}\},{\bm{\Theta}}^{\U},{\bm{\Theta}}^{\D}\right)$ is a solution to problem $\mathscr{P}$, then $\left({\A}, \{{\bm{v}}_{nk}^{\U}\}, \{{p}_k^{\U}\},\{{\bm{v}}_{nk}^{\D}\},{\bm{\Theta}}^{\U},{\bm{\Theta}}^{\D}\right)$ with ${\A}_n=\{k|{\bm{v}}_{nk}\neq\bm{0},~k\in\K\},~\forall\,n\in\N$ is also a solution to problem $\mathscr{P}_{\textrm{original}}$ achieving the same objective value, and vice versa.

Solving problem $\mathscr{P}$ calls for a joint design of beamforming vectors, transmit power, and phase shifts. As the optimization variables are coupled with each other in constraints \eqref{cons:sinr_downlink} and \eqref{cons:sinr_uplink}, in the next section, we shall propose a BSO approach to effectively decouple the optimization variables.

\section{Block-Structured Optimization Approach} \label{sec: alternating approach}
In this section, we propose a BSO approach to decouple the optimization variables. The main idea of the BSO approach is to partition the variables into several blocks, and alternatively optimize one of the blocks in each iteration while keeping the others fixed \cite{hong2015unified}. Specifically, we partition the five optimization variables into three blocks, denoted as $\B_1=\left(\{\bm{v}_{nk}^{\D}\}, \{\bm{v}_{nk}^{\U}\}, \{p_k^{\U}\}\right), \B_2=\bm{\Theta}^{\D}$, and $\B_3=\bm{\Theta}^{\U}$. The main motivation for such a partition is that for given $\bm{\Theta}^{\U}$ and $\bm{\Theta}^{\D}$, problem $\mathscr{P}$ is reduced to a well-studied joint uplink and downlink power minimization problem.

\subsection{Optimizing Variables \texorpdfstring{$\{\bm{v}_{nk}^{\U}\}, \{p_k^{\U}\}$}{text}, and \texorpdfstring{$\{\bm{v}_{nk}^{\D}\}$}{text}} \label{section: subproblem 1}
When $\B_2$ and $\B_3$ are fixed, problem $\mathscr{P}$ is reduced to the following problem
\begin{eqnarray}\label{prob:sub1 optimize v}
\mathscr{P}_\textrm{1}: &\underset{\{\bm{v}_{nk}^{\D}\}, \{\bm{v}_{nk}^{\U}\}, \{p_k^{\U}\}}{\mathrm{minimize}} ~~&P_{\textrm{total}}(\{p_k^{\U}\}, \{\bm{v}_{nk}^{\D}\})\nonumber\\
&\mathrm{subject~to} & \eqref{cons:sinr_downlink}-\eqref{cons:power_uplink}. \nonumber
\end{eqnarray}
%\begin{eqnarray}\label{prob:sub1 optimize v}
%\mathscr{P}_\textrm{1}: &\underset{\{\bm{v}_k^{\D}\}, \{\bm{v}_k^{\U}\}, \{p_k^{\U}\}}{\mathrm{minimize}} ~~&P_{\textrm{total}}\nonumber\\
%&\mathrm{subject~to} ~~&\textrm{SINR}_k^{\U} \geq \gamma_k^{\U}, ~~\forall k, \nonumber\\
%&\quad \quad \quad \quad \quad	&\textrm{SINR}_k^{\D} \geq \gamma_k^{\D}, ~~\forall k, \nonumber\\
%& \quad \quad \quad \quad \quad &p_k^{\U} \leq P_{k,\textrm{max}}^{\U}, ~~\forall k, \nonumber\\
%& \quad \quad \quad \quad \quad &\sum_{k=1}^{K} \left\|\bm{v}_{nk}^{\D} \right\|_2^2 \leq P_{n,\textrm{max}}^{\D}, ~~\forall n. \nonumber
%\end{eqnarray}
It is observed that $\bm{v}_{nk}^{\U}$ and $\bm{v}_{nk}^{\D}$ are only coupled in the objective function. For the sake of analysis convenience and efficient algorithm design, we temporarily dismiss the indicator function in the objective function and split problem $\mathscr{P}_\textrm{1}$ into two parts, i.e., the downlink part $\mathscr{P}_\textrm{1-1}$ and the uplink part $\mathscr{P}_\textrm{1-2}$.

Specifically, the power minimization problem in the downlink part is 
\begin{eqnarray}
\mathscr{P}_\textrm{1-1}: &&\underset{\{\bm{v}_{nk}^{\D}\}}{\mathrm{minimize}} ~~\sum_{n=1}^{N} \sum_{k=1}^{K}\frac{1}{\eta_n} \left\|\bm{v}_{nk}^{\D} \right\|_2^2\nonumber\\
&&\mathrm{subject~to}~~ \eqref{cons:sinr_downlink},\eqref{cons:power_downlink},\nonumber
\end{eqnarray}
which is a celebrated problem formulation in unicast multiple-input single-output (MISO) systems. Due to the fact that an arbitrary phase rotation of transmit beamforming vectors does not affect the SINR constraints \eqref{cons:sinr_downlink} \cite{wiesel2006linear}, we can replace \eqref{cons:sinr_downlink} with the following second-order cone (SOC) constraint
\begin{eqnarray}\label{cons: downlink soc}
 \sqrt{\sum_{l\neq k}\left|\left(\bm{g}_k^{\D}\right)^{\H}\bm{v}_{l}^{\D}\right|^2+\sigma_{k}^2} \leq \frac{1}{\sqrt{\gamma_k^{\D}}}\mathfrak{R}\left(\left(\bm{g}_k^{\D}\right)^{\H}\bm{v}_{k}^{\D}\right),
\end{eqnarray} 
where 
$
\bm{g}_k^{\D}=\left[\left(\bm{g}_{1k}^{\D}\right)^{\sf{T}},\dots,\left(\bm{g}_{Nk}^{\D}\right)^{\sf{T}}\right]^{\sf{T}}$ and $ \bm{v}_k^{\D}=\left[\left(\bm{v}_{1k}^{\D}\right)^{\sf{T}},\dots,\left(\bm{v}_{Nk}^{\D}\right)^{\sf{T}}\right]^{\sf{T}}$
denote the aggregated channel response vector and transmit beamforming vector with respect to MD $k$, respectively.
Therefore, problem $\mathscr{P}_\textrm{1-1}$ is recast
%\begin{subequations}\label{prob:convex downlink subproblem}
%\begin{eqnarray}
%&&\underset{\{\bm{v}_{nk}^{\D}\}}{\mathrm{minimize}} ~~\sum_{n=1}^{N} \sum_{k=1}^{K} \frac{1}{\eta_n}\left\|\bm{v}_{nk}^{\D} \right\|_2^2\nonumber\\
%&&\mathrm{subject~to} ~~\eqref{cons:power_downlink},\nonumber\\
%&& \sqrt{\sum_{l\neq k}\left|\left(\bm{h}_k^{\D}\right)^{\H}\bm{v}_{l}^{\D}\right|^2+\sigma_{k}^2} \leq \frac{1}{\sqrt{\gamma_k^{\D}}}\left(\bm{h}_k^{\D}\right)^{\H}\bm{v}_{k}^{\D}, ~~ \label{cons:downlink_part1}\\
%&&\left(\bm{h}_k^{\D}\right)^{\H}\bm{v}_{k}^{\D} = \mathfrak{R}\left(\left(\bm{h}_k^{\D}\right)^{\H}\bm{v}_{k}^{\D}\right), \forall\,k\in\K,\label{cons:downlink_part2}
%%&& \left\|\begin{array}{c}{\left(\bm{h}_k^{\D}\right)^{{\H}} \bm{V}^{\D}} \\ {\sigma}_k\end{array}\right\| \leq \sqrt{1+\frac{1}{\gamma_{k}^{\D}}} \left(\bm{h}_k^{\D}\right)^{\H}\bm{v}_{k}^{\D}, \forall\,k \in \mathcal{K}
%\end{eqnarray}
%\end{subequations}
as a convex second-order cone programming (SOCP), which can be effectively solved by interior-point methods using modern software like CVX \cite{grant2014cvx}.

On the other hand, the power minimization problem in the uplink is given by
\begin{eqnarray}\label{prob:uplink power minimization problem}
\mathscr{P}_\textrm{1-2}: &&\underset{\{\bm{v}_{nk}^{\U}\}, \{p_k^{\U}\}}{\mathrm{minimize}} ~~\sum_{k=1}^{K} p_k^{\U}\nonumber\\
&&\mathrm{subject~to} ~~\eqref{cons:sinr_uplink}, \eqref{cons:power_uplink}.\nonumber
%&& \quad \quad \quad \quad \quad p_k^{\U} \leq P_{k,\textrm{max}}^{\U},  ~~\forall k. \nonumber
%&& \quad \quad \quad \quad \quad p_k^{\U} \geq 0, ~~\forall k.
\end{eqnarray}
Although the SINR constraints \eqref{cons:sinr_uplink} are similar to those in the downlink counterpart, we cannot convexify them in a similar way. The main reason is that the phase rotation of the receive beamforming vector $\bm{v}_{nk}^{\U}$ cannot guarantee both numerator and denominator of the uplink $\textrm{SINR}$ expression defined in \eqref{eq: uplink sinr} to be real numbers. Moreover, another issue in $\mathscr{P}_\textrm{1-2}$ is that directly optimizing this problem makes $\bm{v}_{nk}^{\U}$ to be nearly zero, because an arbitrary scaling of $\bm{v}_{nk}^{\U}$ does not affect the uplink SINR constraints \eqref{cons:sinr_uplink} \cite{luo2014downlink}. Specifically, if $\left(\tilde{\bm{v}}_{nk}^{\U}, \tilde{p}_k^{\U}\right)$ denotes the optimal solution to problem $\mathscr{P}_{\textrm{1-2}}$, then we have
\begin{eqnarray}\label{eq: scaling issue}
\tilde{\bm{v}}_{nk}^{\U}\approx \bm{0},~\forall\,n\in\N,~\,k\in\K.
\end{eqnarray}
Although \eqref{eq: scaling issue} does not violate the group sparsity structure of $\bm{v}_{nk}$, it indicates that the receive beamforming vectors do not contribute to the task selection. Based on the uplink-downlink duality, in the following, we shall propose a virtual downlink formulation to overcome the scaling issue.

To facilitate effective algorithm design, we relax problem $\mathscr{P}_\textrm{1-2}$ to the following problem 
\begin{subequations}\label{prob: relaxed uplink primal}
\begin{eqnarray}
\underset{\{\bm{v}_{nk}^{\U}\}, \{p_{nk}\}}{\mathrm{minimize}} && \frac{1}{N}\sum_{n=1}^{N}\sum_{k=1}^{K} p_{nk} \nonumber \\
\mathrm{subject~to} 
&& p_{nk}\leq P_{k,\textrm{max}}^{\U}, ~\forall\,n, \forall\,k, \label{cons:relaxed power uplink}\\
&& \hspace{-8em}\frac{p_{nk}\left|\left(\boldsymbol{v}_{n k}^{\U}\right)^{\mathrm{H}} \boldsymbol{g}_{n k}^{\U}\right|^{2}}{\sum_{l \neq k} p_{nl}\left|\left(\boldsymbol{v}_{n k}^{\U}\right)^{\mathrm{H}} \boldsymbol{g}_{n l}^{\U}\right|^{2}\!+\!\sigma_n^{2}\left\|\boldsymbol{v}_{n k}^{\U}\right\|_{2}^{2}} \geq \gamma_{k}^{\U}, \forall\,n, \forall\,k. \label{cons:relaxed sinr uplink}
\end{eqnarray}
\end{subequations}
We can easily verify that problem \eqref{prob: relaxed uplink primal} is indeed a relaxation to problem $\mathscr{P}_\textrm{1-2}$, because given any solution $\left(\{\bm{v}_{nk}^{\U}\}, \{p_{k}^{\U}\}\right)$ feasible to problem $\mathscr{P}_\textrm{1-2}$, $\left(\{\bm{v}_{nk}^{\U}\}, \{p_{nk}\}\right)$ is also feasible to problem $\eqref{prob: relaxed uplink primal}$ when $p_{nk}=p_{k}^{\U}, \forall\,n\in\N$, and the objective value of problem \eqref{prob: relaxed uplink primal} is no greater than that of $\mathscr{P}_\textrm{1-2}$. The motivation for this relaxation is that problem \eqref{prob: relaxed uplink primal} can be solved in the virtual downlink formulation so as to overcome the scaling issue.

We first consider an ideal scenario that the MDs have unlimited power budgets (i.e., $P_{k,\textrm{max}}^{\U}=+\infty, \forall\,k\in\K$). Problem \eqref{prob: relaxed uplink primal} is then equivalent to the following virtual downlink power minimization problem
%
%Specifically, it was shown in \cite{schubert2005iterative} that $\mathscr{P}_\textrm{1-2}^{\prime}$ is equivalent to the following \textit{virtual} downlink power minimization problem 
\begin{subequations}\label{prob:virtual downlink power minimization without mild constraint}
\begin{eqnarray}
\underset{\{\bm{v}_{nk}^{\V}\}}{\mathrm{minimize}} &&\frac{1}{N}\sum_{n=1}^{N}\sum_{k=1}^{K} \left\|\bm{v}_{nk}^{\V}\right\|_2^2 \label{placeholder}\\
\mathrm{subject~to} &&\textrm{SINR}_{nk}^{\V} \geq \gamma_k^{\U}, ~~\forall\,n, \forall\,k, \label{cons:sinr_virtual_downlink}
\end{eqnarray}
\end{subequations}
where $\bm{v}_{nk}^{\V}\in\C^{L_n\times 1}$ denotes the virtual downlink transmit beamforming vector from BS $n$ to MD $k$, and $\textrm{SINR}_{nk}^{\V}$ is the virtual downlink SINR observed by MD $k$ defined as 
\begin{eqnarray}\label{eq:sinr vdl}
\textrm{SINR}_{nk}^{\V} = \frac{\left|\left(\bm{g}_{nk}^{\U}\right)^{\H}\bm{v}_{nk}^{\V}\right|^2}{\sum_{l\neq k}\left|\left(\bm{g}_{nk}^{\U}\right)^{\H}\bm{v}_{nl}^{\V}\right|^2 + \sigma_n^2}~.
\end{eqnarray}
It is noticed that in \eqref{eq:sinr vdl}, the scaling issue no longer exists for $\{\bm{v}_{nk}^{\V}, n\in\N, k\in\K\}$. 
The rigorous proof of the equivalence of the uplink power minimization problem \eqref{prob: relaxed uplink primal} and the virtual downlink problem \eqref{prob:virtual downlink power minimization without mild constraint} can be derived by Lagrangian duality, as in \cite[Theorem 1]{dahrouj2010coordinated}.
The optimal solutions obtained by solving problem \eqref{prob:virtual downlink power minimization without mild constraint} have close connections to solutions to problem \eqref{prob: relaxed uplink primal}, i.e., $\bm{v}_{nk}^{\V} = \bm{v}_{nk}^{\U}$ and $\sum_{n=1}^{N}\sum_{k=1}^{K}  \left\|\bm{v}_{nk}^{\V}\right\|_2^2= \sum_{n=1}^{N}\sum_{k=1}^{K} p_{nk}$. However, it is worth mentioning that the equivalence between the virtual downlink beamforming power and the uplink transmit power does not necessarily hold, i.e., $\left\|\bm{v}_{nk}^{\V}\right\|_2^2 \neq p_{nk}$. Therefore if $P_{k,\textrm{max}}^{\U}<+\infty$, we cannot directly rewrite the transmit power constraints \eqref{cons:relaxed power uplink} as $\|\bm{v}_{nk}^{\V}\|_2^2 \leq P_{k,\textrm{max}}^{\U}, \forall\,n, \forall\,k$ out of intuition and add them to problem \eqref{prob:virtual downlink power minimization without mild constraint}. Instead, we consider a sum-power constraint to relax the uplink transmit power constraints \eqref{cons:relaxed power uplink}, which is given by 
\begin{eqnarray}\label{constraint:relaxed sum transmit power}
	\sum_{n=1}^{N}\sum_{k=1}^{K}  \left\|\bm{v}_{nk}^{\V}\right\|_2^2= \sum_{n=1}^{N}\sum_{k=1}^{K} p_{nk}\leq N\sum_{k=1}^{K}P_{k,\textrm{max}}^{\U}.
\end{eqnarray}
By introducing this mild power control to problem \eqref{prob:virtual downlink power minimization without mild constraint},
%We can easily verify that constraint \eqref{constraint:relaxed sum transmit power} is indeed a relaxation to \eqref{cons:power_uplink}, because the feasibility of \eqref{cons:power_uplink} guarantees the feasibility of \eqref{constraint:relaxed sum transmit power}, but the opposite is not true.
we need to solve the following problem
\begin{eqnarray}\label{prob:relaxed uplink power minimization problem} 
&&\underset{\{\bm{v}_{nk}^{\V}\}}{\mathrm{minimize}} ~~\frac{1}{N}\sum_{n=1}^{N}\sum_{k=1}^{K} \left\|\bm{v}_{nk}^{\V}\right\|_2^2\nonumber\\
&&\mathrm{subject~to} ~~\eqref{cons:sinr_virtual_downlink}, \eqref{constraint:relaxed sum transmit power}.
\end{eqnarray}
Similar to \eqref{cons: downlink soc}, constraint \eqref{cons:sinr_virtual_downlink} can be replaced with the following SOC constraint
\begin{eqnarray}
\sqrt{\sum_{l\neq k}\left|\left(\bm{g}_{nk}^{\U}\right)^{\H}\bm{v}_{nl}^{\V}\right|^2+\sigma_n^2} \leq \frac{1}{\sqrt{\gamma_k^{\V}}}\mathfrak{R}\left(\left(\bm{g}_{nk}^{\U}\right)^{\H}\bm{v}_{nk}^{\V}\right) \label{cons: virtual downlink soc}
\end{eqnarray}
to make problem \eqref{prob:relaxed uplink power minimization problem} a convex SOCP problem.
%\begin{subequations}\label{prob:convex uplink subproblem} 
%\begin{eqnarray}
%&&\underset{\{\bm{v}_k^{\V}\}}{\mathrm{minimize}} ~~\sum_{k=1}^{K} \left\|\bm{v}_k^{\V}\right\|_2^2\nonumber\\
%&&\mathrm{subject~to} ~~\eqref{constraint:relaxed sum transmit power}, \nonumber\\
%&&\\
%&& \left(\bm{h}_k^{\V}\right)^{\H}\bm{v}_{k}^{\V} = \mathfrak{R}\left(\left(\bm{h}_k^{\V}\right)^{\H}\bm{v}_{k}^{\V}\right), \forall\,k.\label{cons:virtual_downlink_part2}
%\end{eqnarray}
%\end{subequations}
%That is, denote the optimal solution to problem \eqref{prob:uplink power minimization problem} as $\{\tilde{\bm{v}}_k^{\U}\}$ and $\{\tilde{p}_k^{\U}\}$, we can always find a feasible solution to problem \eqref{prob:virtual downlink power minimization without mild constraint} denoted as $\{\tilde{\bm{v}}_k^{\V}\}$ with the same objective value. However, the opposite direction might not hold because of the relaxed MD constraint. 
%Furthermore, on solving \eqref{prob:relaxed uplink power minimization problem}, $\{\bm{v}_k^{\U}\}$ in $\mathscr{P}_\textrm{1-2}$ are obtained by setting $\bm{v}_k^{\U}=\bm{v}_k^{\V}, \forall k$. Uplink transmit power $\{p_k\}$ are obtained by solving $\mathscr{P}_\textrm{1-2}$ with $\bm{v}_k^{\U}$ fixed as $\bm{v}_k^{\U}=\bm{v}_k^{\V}, \forall k$, then $\mathscr{P}_\textrm{1-2}$ is reduced to a linear programming and can be efficiently solved by CVX.

By exploiting the uplink-downlink duality and transforming the uplink model into a virtual downlink model, we address the challenge of the receive beamforming vectors scaling issue. As $\{\bm{v}_{nk}^{\V}\}$ have the same group sparsity pattern as $\{\bm{v}_{nk}^{\U}\}$, combining the downlink and virtual downlink parts, we relax $\mathscr{P}_1$ to the following problem
\begin{eqnarray}\label{prob: relaxed virtual power minimization problem}
&&\underset{\{\bm{v}_{nk}^{\D}\}, \{\bm{v}_{nk}^{\V}\}}{\mathrm{minimize}} ~~\frac{1}{N}\sum_{n=1}^{N}\sum_{k=1}^{K}\left\|\bm{v}_{nk}^{\V}\right\|_2^2+  \sum_{n=1}^{N}\sum_{k=1}^{K}  \frac{1}{\eta_n}\left\|\bm{v}_{nk}^{\D} \right\|_2^2 \nonumber \\
&&  \hspace{6em}+ \sum_{n=1}^{N}\sum_{k=1}^{K}\bm{1}_{\{\bm{\bar{v}}_{nk}=\bm{0}\}}P_{nk}^{\sf{c}}\nonumber\\
&&\mathrm{subject~to}~~\eqref{cons:power_downlink}, \eqref{cons: downlink soc}, \eqref{constraint:relaxed sum transmit power}, \eqref{cons: virtual downlink soc},  
\end{eqnarray}
where $\bm{\bar{v}}_{nk}$ in the objective function is defined as $\bm{\bar{v}}_{nk}=\left[\left(\bm{v}_{nk}^{\V}\right)^{{\sf{T}}}, \left(\bm{v}_{nk}^{\D}\right)^{{\sf{T}}}\right]^{{\sf{T}}}$.
%, and the indicatior function is defined as 
%\begin{eqnarray}
%\bm{1}\left(\bm{\bar{v}}_{nk}\right) = \left\{\begin{array}{ll} {0,} & \mathrm{if}~ \bm{\bar{v}}_{nk}=\bm{0} \\ {1,} & \mathrm{otherwise} \end{array} .\right. \nonumber
%\end{eqnarray}
%\begin{eqnarray}\label{eq:indicator}
%\sum_{n=1}^{N}\sum_{k=1}^{K}\bm{1}\left(\bm{v}_{nk}^{\V}, \bm{v}_{nk}^{\D}\right)P_{nk}^{\sf{c}} 
%= \sum_{n=1}^{N}\sum_{k=1}^{K}\bm{1}\left(\bm{v}_{nk}\right)P_{nk}^{\sf{c}},
%\end{eqnarray}
It is observed that all the constraints of problem \eqref{prob: relaxed virtual power minimization problem} are convex. 

Once we solve problem \eqref{prob: relaxed virtual power minimization problem} and obtain the solutions $\left(\{\bm{v}_{nk}^{\D}\}, \{\bm{v}_{nk}^{\V}\}\right)$, the receive beamforming vector $\{\bm{v}_{nk}^{\U}\}$ in $\mathscr{P}_\textrm{1}$ can be obtained by setting $\bm{v}_{nk}^{\U} = \bm{v}_{nk}^{\V}, \forall\,n\in\N, k\in\K$, and $\{p_k^{\U}\}$ can be obtained by solving problem $\mathscr{P}_{\textrm{1-2}}$ with $\{\bm{v}_{nk}^{\U}\}$ fixed, which is reduced to a linear program.

\subsection{Optimizing Variable \texorpdfstring{$\bm{\Theta}^{\D}$}{text} }
For given $\B_1$ and $\B_3$, we optimize $\bm{\Theta}^{\D}$ to solve the resulting problem, which is termed as $\mathscr{P}_\textrm{2}$. Since $\bm{\Theta}^{\D}$ does not appear in the objective function of $\mathscr{P}$, problem $\mathscr{P}_\textrm{2}$ is in fact a downlink feasibility detection problem, which is given by
\begin{eqnarray}\label{prob: optimize theta downlink}
\mathscr{P}_\textrm{2}: &\mathrm{find}~~ &\bm{\Theta}^{\D} \nonumber\\
&\mathrm{subject~to} ~~&\eqref{cons:sinr_downlink}, \eqref{cons:phase_downlink}. \nonumber
%&\quad \quad \quad \quad \quad &0\leq \theta_m^{\D} \leq 2\pi, ~~\forall\,m, 
\end{eqnarray}
According to the $\textrm{SINR}^{\D}$ expression given in \eqref{sinr_downlink}, we have
%\begin{eqnarray}
%\hspace{-3.5em}&&\left| \sum_{n=1}^{N}\left(\bm{g}_{nk}^{\D}\right)^{\H} \boldsymbol{v}_{nl}^{\D}\right|^{2}\stackrel{(a)}{=}\left|\sum_{n=1}^{N}\left(\left(\bm{h}_{\textrm{d},nk}^{\D}\right)^{{\H}}+\left(\bm{h}_{\textrm{r},k}^{\D}\right)^{{\H}}\bm{\Theta}^{\D}\bm{G}_n^{\D}\right)\bm{v}_{nl}^{\D}\right|^2 \nonumber\\
%\hspace{-3.5em}&&\stackrel{(b)}{=}\left|\left(\bm{h}_{\textrm{d},k}^{\D}\right)^{{\H}}\bm{v}_l^{\D}+\beta\left(\bm{a}^{\D}\right)^{{\H}}\mathrm{diag}\left(\left(\bm{h}_{\textrm{r},k}^{\D}\right)^{{\H}}\right)\tilde{\bm{G}}^{\D}\bm{v}_l^{\D}\right|^2, \label{eq:b}
%\end{eqnarray}
\begin{eqnarray}
&&\hspace{-2.2em}\left| \sum_{n=1}^{N}\left(\bm{g}_{nk}^{\D}\right)^{\H} \boldsymbol{v}_{nl}^{\D}\right|^{2}  \nonumber \\
&&\hspace{-3.5em}\stackrel{(a)}{=}\left|\sum_{n=1}^{N}\left(\left(\bm{h}_{\textrm{d},nk}^{\D}\right)^{{\H}}+\left(\bm{h}_{\textrm{r},k}^{\D}\right)^{{\H}}\bm{\Theta}^{\D}\bm{G}_n^{\D}\right)\bm{v}_{nl}^{\D}\right|^2 \nonumber\\
&&\hspace{-3.5em}\stackrel{(b)}{=}\left|\left(\bm{h}_{\textrm{d},k}^{\D}\right)^{{\H}}\bm{v}_l^{\D}+\beta\left(\bm{a}^{\D}\right)^{{\H}}\mathrm{diag}\left(\left(\bm{h}_{\textrm{r},k}^{\D}\right)^{{\H}}\right)\tilde{\bm{G}}^{\D}\bm{v}_l^{\D}\right|^2, \label{eq:b}
\end{eqnarray}
where $(a)$ follows by substituting \eqref{eq: downlink channel}, and $(b)$ holds by defining
\begin{eqnarray}
\bm{h}_{\textrm{d},k}^{\D} = \left[\left(\bm{h}_{\textrm{d},1k}^{\D}\right)^{{\sf{T}}},\dots, \left(\bm{h}_{\textrm{d},Nk}^{\D}\right)^{{\sf{T}}}\right]^{{\sf{T}}},  \bm{a}^{\D}=\left[\theta_1^{\D},\dots, \theta_M^{\D}\right]^{{\H}}, \nonumber
\end{eqnarray}
\vspace{-4ex}
\begin{eqnarray}
\bm{v}_l^{\D}= \left[\left(\bm{v}_{1l}^{\D}\right)^{{\sf{T}}},\dots, \left(\bm{v}_{Nl}^{\D}\right)^{{\sf{T}}}\right]^{{\sf{T}}}, \tilde{\bm{G}}^{\D}=\left[\bm{G}_1^{\D},\dots,\bm{G}_N^{\D}\right]. \nonumber
\end{eqnarray}
%\setlength\belowdisplayskip{+0.05pt} 
%\vspace{-3ex}
%\setlength\abovedisplayskip{+2.5pt} 
Note that in \eqref{eq:b} the only term related to phase shifts is $\bm{a}^{\D}$. Therefore for notational ease, we define 
\vspace{-1ex}
\begin{eqnarray}
\bm{w}_{kl}^{\D} =\beta \mathrm{diag}\left(\left(\bm{h}_{\textrm{r},k}^{\D}\right)^{{\H}}\right)\tilde{\bm{G}}^{\D}\bm{v}_l^{\D}, 
\quad b_{kl}^{\D} =\left(\bm{h}_{\textrm{d},k}^{\D}\right)^{{\H}}\bm{v}_l^{\D},  \nonumber
\end{eqnarray}
and the $\textrm{SINR}^{\D}$ expression in \eqref{sinr_downlink} can be equivalently rewritten as
\begin{eqnarray}
\textrm{SINR}_k^{\D} = \frac{\left|b_{kk}^{\D}+\left(\bm{a}^{\D}\right)^{{\H}}\bm{w}_{kk}^{\D}\right|^2}{\sum_{l \neq k}\left|b_{kl}^{\D}+\left(\bm{a}^{\D}\right)^{{\H}}\bm{w}_{kl}^{\D}\right|^2 + \sigma_{k}^2},
\end{eqnarray}
leading to the following inhomogeneous QCQP problem
\begin{subequations}\label{prob:sub2original}
\begin{eqnarray}
&\mathrm{find} ~&\bm{a}^{\D} \nonumber\\
&\mathrm{subject~ to~} & \left|a_m^{\D}\right|^2=1, \forall\,m,\\
&&\hspace{-4em}\frac{\left|b_{kk}^{\D}+\left(\bm{a}^{\D}\right)^{{\H}}\bm{w}_{kk}^{\D}\right|^2}{\sum_{l \neq k}\left|b_{kl}^{\D}+\left(\bm{a}^{\D}\right)^{{\H}}\bm{w}_{kl}^{\D}\right|^2 + \sigma_{k}^2} \geq \gamma_k, \forall\, k.
\end{eqnarray}
\end{subequations}
By introducing an auxiliary scalar $t$, and defining
\begin{equation}\label{eq:R, a}
\boldsymbol{R}_{kl}^{\D}=\left[ \begin{array}{cc}{\boldsymbol{w}_{kl}^{\D} \left(\boldsymbol{w}_{kl}^{\D}\right)^{{\H}}} & {\boldsymbol{w}_{kl}^{\D} \left(b_{kl}^{\D}\right)^{{\H}}} \\ {\left(\boldsymbol{w}_{kl}^{\D}\right)^{{\H}} b_{kl}^{\D}} & {0}\end{array}\right], \nonumber
\bm{\bar{a}}^{\D}=\left[ \begin{array}{c}{\bm{a}^{\D}} \\ {t}^{\D}\end{array}\right],
\end{equation}
problem \eqref{prob:sub2original} can be transformed into the following homogeneous QCQP problem \cite{wu2018intelligent}
\begin{eqnarray}\label{prob:DL qcqp}
&\mathrm{find} &\bm{\bar{a}}^{\D} \nonumber\\
&\mathrm{subject~to~} &\left|\bar{a}_m^{\D}\right|^2=1, ~~\mathrm{for}~m=1,\dots,M+1,\nonumber \\
& &\hspace{-4.7em}\frac{\left(\bm{\bar{a}}^{\D}\right)^{{\H}}\bm{R}_{kk}^{\D}\bm{\bar{a}}^{\D} + \left|b_{kk}^{\D}\right|^2}{\sum_{l \neq k}\left(\bm{\bar{a}}^{\D}\right)^{{\H}}\bm{R}_{kl}^{\D}\bm{\bar{a}}^{\D} + \left|b_{kl}^{\D}\right|^2+\sigma_{k}^2} \geq \gamma_k, \forall\,k.
\end{eqnarray}

%Both problems \eqref{prob:DL qcqp} and \eqref{prob:UL qcqp} can be unified into the following general form 
%\begin{eqnarray}\label{prob:general qcqp}
%&\mathrm{find} &\bm{\bar{a}} \nonumber\\
%&\mathrm{subject~to~} &\frac{\bm{\bar{a}}^{{\H}}\bm{R}_{kk}\bm{\bar{a}} + \left|b_{kk}\right|^2}{\sum_{l \neq k}\alpha_{kl}\left(\bm{\bar{a}}^{{\H}}\bm{R}_{kl}\bm{\bar{a}} + \left|b_{kl}\right|^2\right)+c_k} \geq \gamma_k, \forall k, \nonumber\\
%& &\left|\bar{a}_m\right|^2=1, ~~\mathrm{for}~m=1,\dots,M+1,
%\end{eqnarray}
%where $\alpha_{kl}$ and $c_{k}$ are non-negative coefficients defined as
%\begin{eqnarray}\label{eq: alpha, beta}
%	\left\{\begin{array}{ll} {\alpha_{kl} = 1~,~c_{k} = \sigma_{k}^2,} & \textrm{for downlink}\\ 
%	{\alpha_{kl} = {p_l^{\U}}/{p_k^{\U}}~,~c_{k}= \frac{\sigma^2\left\|\bm{v}_k^{\U}\right\|_2^2}{p_k^{\U}},} & \textrm{for uplink} \end{array} .\right.
%\end{eqnarray}
The phase-shift matrix $\bm{\Theta}^{\D}$ can be recovered from $\bm{\bar{a}}^{\D}$ as follows. After solving problem \eqref{prob:DL qcqp} and obtaining a feasible solution $\bm{\bar{a}}^{\D}=\left[\bm{a}_0^{\D}~t_0^{\D}\right]^{{\sf{T}}}$, the solution to problem \eqref{prob:sub2original} can be computed as $\bm{a}^{\D} = \bm{a}_0^{\D}/t_0^{\D}$, and the solution to $\mathscr{P}_\textrm{2}$ is given as $\bm{\Theta}^{\D}=\beta\mathrm{diag}\left( \left(\bm{a}^{\D}\right)^{{\H}}\right)$.

\subsection{Optimizing Variable \texorpdfstring{$\bm{\Theta}^{\U}$}{text}}
As  $\bm{\Theta}^{\U}$ does not appear in the objective function of $\mathscr{P}$, given $\B_1$ and $\B_2$, the resulting problem denoted as $\mathscr{P}_\textrm{3}$ is also a feasibility detection problem
\begin{eqnarray}\label{prob: optimize theta uplink}
\mathscr{P}_\textrm{3}: &\mathrm{find}~~ &\bm{\Theta}^{\U} \nonumber\\
&\mathrm{subject~to} ~~&\eqref{cons:sinr_uplink}, \eqref{cons:phase_uplink}. \nonumber
%&\quad \quad \quad \quad \quad &0\leq \theta_m^{\U} \leq 2\pi, ~~\forall\,m.
\end{eqnarray}
The same derivation process presented in the last subsection is also applicable to transform the uplink problem $\mathscr{P}_\textrm{3}$ into a homogeneous QCQP, therefore details are omitted here for brevity. Specifically, by defining
\begin{eqnarray}
\bm{w}_{nkl}^{\U} = \beta \mathrm{diag}\left(\left(\bm{h}_{\textrm{r},l}^{\U}\right)^{{\H}}\right)\bm{G}_n^{\U}\bm{v}_{nk}^{\U}, \quad b_{nkl} = \left(\bm{h}_{\textrm{d},{nl}}^{\U}\right)^{{\H}}\bm{v}_{nk}^{\U},  \nonumber
\end{eqnarray}
\begin{equation}
\boldsymbol{R}_{nkl}^{\U}=\left[ \begin{array}{cc}{\boldsymbol{w}_{nkl}^{\U} \left(\boldsymbol{w}_{nkl}^{\U}\right)^{{\H}}} & {\boldsymbol{w}_{nkl}^{\U} \left(b_{nkl}^{\U}\right)^{{\H}}} \\ {\left(\boldsymbol{w}_{nkl}^{\U}\right)^{{\H}} b_{nkl}^{\U}} & {0}\end{array}\right], ~~ \bm{\bar{a}}^{\U}=\left[ \begin{array}{c}{\bm{a}^{\U}} \\ {t}^{\U}\end{array}\right], \nonumber
\end{equation}
we have the following uplink homogeneous QCQP problem
\begin{eqnarray}\label{prob:UL qcqp}
\mathrm{find}\quad~\,&&\bm{\bar{a}}^{\U} \nonumber\\
\mathrm{subject~to}&&\left|\bar{a}_m^{\U}\right|^2=1, ~~\mathrm{for}~m=1,\dots,M+1, \nonumber\\
&&\hspace{-7em}\frac{p_k^{\U}\left(\left(\bm{\bar{a}}^{\U}\right)^{{\H}}\bm{R}_{nkk}^{\U}\bm{\bar{a}}^{\U} + \left|b_{nkk}^{\U}\right|^2\right)}{\sum_{l \neq k}p_l^{\U}\left(\left(\bm{\bar{a}}^{\U}\right)^{{\H}}\bm{R}_{nkl}^{\U}\bm{\bar{a}}^{\U} + \left|b_{nkl}^{\U}\right|^2\right)+\sigma_n^2\left\|\bm{v}_{nk}^{\U}\right\|_2^2} \geq \gamma_k, \nonumber\\
&& \hspace{13em}\forall\,n, \forall\,k.
\end{eqnarray}

\subsection{A Unified BSO Approach}
Based on the above discussions, problem $\mathscr{P}$ is solved by iteratively solving problems $\mathscr{P}_\textrm{1}$, $\mathscr{P}_\textrm{2}$, and $\mathscr{P}_\textrm{3}$ in an alternating manner until convergence.
%\begin{itemize}
%	\item the improvement of consecutive objective values in \eqref{prob: sub1_convex} is below a predefined threshold $\epsilon$;
%	\item problem \eqref{prob:sub2matrixlifting} becomes infeasible.
%\end{itemize}
We justify the effectiveness and depict the convergence behavior of the proposed BSO approach in the following proposition.
\begin{proposition}\label{proposition: convergence}
	With the BSO approach, the objective value of $\mathscr{P}_\textrm{1}$ is non-increasing in the consecutive iterations.
\end{proposition}
\begin{proof}
	For ease of notation, we denote the objective value of $\mathscr{P}_\textrm{1}$ as $f(\bm{V}, \bm{\Omega})$, where the first variable $\bm{V}$ is an abstraction of three optimization variables $\{\bm{v}_{nk}^{\D}\}, \{\bm{v}_{nk}^{\U}\},$ and $\{p_k^{\U}\}$ in $\mathscr{P}_\textrm{1}$, and the second variable $\bm{\Omega}$ abstracts phase-shift matrices $\bm{\Theta}^{\D}$ and $\bm{\Theta}^{\U}$ in $\mathscr{P}_\textrm{2}$ and $\mathscr{P}_\textrm{3}$. Assuming that $(\bm{V}^{(t)}, \bm{\Omega}^{(t)})$ is obtained at iteration $t$. If $\mathscr{P}_\textrm{2}$ and $\mathscr{P}_\textrm{3}$ are feasible, i.e., $(\bm{V}^{(t)}, \bm{\Omega}^{(t+1)})$ exists, $(\bm{V}^{(t)}, \bm{\Omega}^{(t+1)})$ is also feasible to $\mathscr{P}_\textrm{1}$. Therefore, $(\bm{V}^{(t)}, \bm{\Omega}^{(t)})$ and $(\bm{V}^{(t+1)}, \bm{\Omega}^{(t+1)})$ are feasible solutions to $\mathscr{P}_\textrm{1}$ at the consecutive iterations $t$ and $t+1$, respectively. We have the following inequality
	\begin{eqnarray}
	\quad f(\bm{V}^{(t+1)}, \bm{\Omega}^{(t+1)}) \stackrel{(a)}{\leq} f(\bm{V}^{(t)}, \bm{\Omega}^{(t+1)})  \stackrel{(b)}{=} f(\bm{V}^{(t)}, \bm{\Omega}^{(t)}), \nonumber
	\end{eqnarray}
	where $(a)$ holds because $\bm{V}^{(t+1)}$ is the optimal solution to $\mathscr{P}_\textrm{1}$ for a given $\bm{\Omega}^{(t+1)}$ at iteration $t+1$ , and $(b)$ holds because the objective value $P_\textrm{total}$ depends only on the value of $\bm{V}$ and is independent of $\bm{\Omega}$.
\end{proof}

Although the BSO addresses the challenge of coupled optimization variables, problems \eqref{prob: relaxed virtual power minimization problem}, \eqref{prob:DL qcqp}, and \eqref{prob:UL qcqp} are still intractable due to the non-convexity of these problems. In the next section, we shall propose effective algorithms to solve these nonconvex problems.

\section{BSO with Mixed $\ell_{1,2}$-norm and DC Based Three-Stage Framework for Network Power Minimization Problem} \label{sec: Framework}
In this section, we first propose two tractable algorithms for problems \eqref{prob: relaxed virtual power minimization problem},  \eqref{prob:DL qcqp}, and \eqref{prob:UL qcqp}, based on the mixed $\ell_{1,2}$ sparsity inducing norm and DC approach, respectively. Following these two algorithms, a three-stage framework is developed for problem $\mathscr{P}_{\textrm{original}}$.
\subsection{Mixed $\ell_{1,2}$-Norm for Group Sparsity Inducing}\label{subsec:1}
Although the feasible set is convex, problem \eqref{prob: relaxed virtual power minimization problem} is a nonconvex integer programming problem due to the indicator function in the objective function. We identify that the third term in the objective function of problem \eqref{prob: relaxed virtual power minimization problem} is a weighted $\ell_{0}$-norm of vector $\bm{\bar{v}} = \left[\bm{\bar{v}}_{11}^{{\sf{T}}},\bm{\bar{v}}_{12}^{{\sf{T}}},\dots,\bm{\bar{v}}_{NK}^{{\sf{T}}}\right]^{{\sf{T}}}$ with weights $\{P_{nk}^{\sf{c}},\,n\in\N\,k\in\K\}$, and it is non-convex. As $\ell_1$-norm is a well-known convex relaxation to $\ell_{0}$-norm, we relax the weighted $\ell_{0}$-norm as 
\begin{eqnarray}\label{eq: mixed l12 relaxation}
&&\sum_{n=1}^{N}\sum_{k=1}^{K}\bm{1}_{\{\bm{\bar{v}}_{nk}=\bm{0}\}}P_{nk}^{\sf{c}} 
\approx \sum_{n=1}^{N}\sum_{k=1}^{K}\left\|\bm{\bar{v}}_{nk}\right\|_2P_{nk}^{\sf{c}} \nonumber \\
&&\hspace{-1em}=\sum_{n=1}^{N}\sum_{k=1}^{K}\sqrt{\|\bm{v}_{nk}^{\D}\|_2^2 + \|\bm{v}_{nk}^{\V}\|_2^2 }P_{nk}^{\sf{c}}.
\end{eqnarray}
Note that \eqref{eq: mixed l12 relaxation} is actually the weighted mixed $\ell_{1,2}$-norm of vector $\bm{\bar{v}}$.
%In order to induce the group sparsity of $\bm{v}$, the mixed $\ell_{1,2}$-norm is adopted, 
%\begin{eqnarray}
%\left\|\bm{v}\right\|_{\ell_{1,2}} = \sum_{n=1}^{N}\sum_{k=1}^{K}\left\|\left[\bm{v}_{nk}^{\D}, \bm{v}_{nk}^{\V}\right]\right\|_2\nonumber \\
%=\sum_{n=1}^{N}\sum_{k=1}^{K}\sqrt{\|\bm{v}_{nk}^{\D}\|_2^2 + \|\bm{v}_{nk}^{\V}\|_2^2 } ~~.
%\end{eqnarray}
The mixed $\ell_{1,2}$-norm behaves like an $\ell_1$-norm on vector $\left[\|\bm{\bar{v}}_{11}\|_2,\|\bm{\bar{v}}_{12}\|_2, \dots, \|\bm{\bar{v}}_{NK}\|_2\right]$. The outer $\ell_1$-norm induces the sparsity structure, while the inner $\ell_2$-norm is responsible for forcing all coefficients in the beamforming group $\bm{\bar{v}}_{nk}$ to be zero. By adopting mixed $\ell_{1,2}$-norm as the convex relaxation of the indicator function term, we can induce the group sparsity structure of beamforming groups $\{\bm{\bar{v}}_{nk},~n\in\N,\,k\in\K\}$.

\textit{Remark:} In general, all the mixed $\ell_{1,p}$-norm with $p\geq 1$ can induce the group sparsity structure \cite{bach2012optimization}. However for brevity, we adopt the most commonly used mixed $\ell_{1,2}$-norm in this paper.

By replacing the indication function with its convex surrogate, we relax problem \eqref{prob: relaxed virtual power minimization problem} as the following convex problem
\begin{eqnarray}\label{prob: sub1_convex}
&&\underset{\{\bm{v}_{nk}^{\D}\}, \{\bm{v}_{nk}^{\V}\}}{\mathrm{minimize}} ~~\sum_{n=1}^{N}\sum_{k=1}^{K}\|\bm{v}_{nk}^{\V}\|_2^2 +\sum_{n=1}^{N}\sum_{k=1}^{K}\frac{1}{\eta_n} \left\|\bm{v}_{nk}^{\D} \right\|_2^2 \nonumber \\
&& \quad \quad \quad \quad ~+ \sum_{n=1}^{N}\sum_{k=1}^{K}\sqrt{\|\bm{v}_{nk}^{\D}\|_2^2 + \|\bm{v}_{nk}^{\V}\|_2^2 }P_{nk}^{\sf{c}}\nonumber\\
&&\mathrm{subject~to} ~~\eqref{cons:power_downlink}, \eqref{cons: downlink soc}, \eqref{constraint:relaxed sum transmit power}, \eqref{cons: virtual downlink soc}.
\end{eqnarray}

As we have discussed in Section \ref{section: subproblem 1}, the solutions to the non-convex problem $\mathscr{P}_\textrm{1}$ can be obtained by utilizing the solutions to problem \eqref{prob: sub1_convex}. The overall algorithm for solving problem $\mathscr{P}_{\textrm{1}}$ is summarized in Algorithm \ref{algo: sub1}.

As problem $\mathscr{P}_{\textrm{1-2}}$ with fixed $\{\bm{v}_{nk}^{\U}\}$ can be efficiently solved by the fixed-point iterations \cite{wiesel2006linear}, the computational complexity of Algorithm \ref{algo: sub1} is dominated by solving the SOCP problem \eqref{prob: sub1_convex}, which is $\mathcal{O}(L^{3.5}K^{3.5})$ using interior-point methods \cite{boyd2004convex}.

\begin{algorithm}[!h]
	\textbf{Input:} $\bm{\Theta}^{\U}, \bm{\Theta}^{\D}$ \\
	1. solve the convex relaxation problem \eqref{prob: sub1_convex} to obtain $\{\bm{v}_{nk}^{\D}\}, \{\bm{v}_{nk}^{\V}\}$\\
	2. set $\bm{v}_{nk}^{\U} = \bm{v}_{nk}^{\V}$, $\forall\,n\in\N, k\in\K$, and then solve problem $\mathscr{P}_\textrm{1-2}$ with fixed $\{\bm{v}_{nk}^{\U}\}$ to obtain $\{p_k^{\U}\}$\\
	\textbf{Output}: $\{\bm{v}_{nk}^{\D}\}, \{\bm{v}_{nk}^{\U}\}, \{p_k^{\U}\}$
	\caption{Mixed $\ell_{1,2}$-Norm Based Group Sparsity Inducing for Problem $\mathscr{P}_\textrm{1}$}
	\label{algo: sub1}
\end{algorithm}

\subsection{DC Approach for Nonconvex QCQP Problem}\label{subsec:2}
The non-convexity of problems \eqref{prob:DL qcqp} and \eqref{prob:UL qcqp} lie in the unit-modulus constraints. A common technique used to handle the nonconvex QCQP problems is matrix lifting. For ease of notation, we omit the superscripts $\D$ and $\U$ if it does not cause any ambiguity. In problem \eqref{prob:DL qcqp}, as
$
\bm{\bar{a}}^{{\H}}\bm{R}_{kk}\bm{\bar{a}}=\mathrm{Tr}\left(\bm{R}_{kk}\bm{\bar{a}}\bm{\bar{a}}^{{\H}}\right), \nonumber 
$
and by introducing a new variable $\bm{A}=\bm{\bar{a}}\bm{\bar{a}}^{{\H}}\in\mathbb{C}^{(M+1)\times(M+1)}$, we rewrite problem \eqref{prob:DL qcqp} as the following feasibility detection problem 
\begin{subequations}\label{prob:sub2matrixlifting}
\begin{eqnarray}
&\mathrm{find} ~ &\bm{A} \nonumber\\
%&\mathrm{subject~to~} &\frac{\mathrm{Tr}\left(\bm{R}_{kk}\bm{A}\right) + \left|b_{kk}\right|^2}{\sum_{l \neq k}\alpha_{kl}\left(\mathrm{Tr}\left(\bm{R}_{kl}\bm{A}\right) + \left|b_{kl}\right|^2\right)+c_{k}} \geq \gamma_k,  \nonumber\\
&\mathrm{subject~to~} &\mathrm{Tr}\left(\bm{R}_{kk}\bm{A}\right) + \left|b_{kk}\right|^2 \geq \gamma_k \sum_{l \neq k}^{K}\mathrm{Tr}\left(\bm{R}_{kl}\bm{A}\right) \nonumber \\
& &~~+ \gamma_k \left(\sum_{l \neq k}^{K}\left|b_{kl}\right|^2 + \sigma_k^2\right), \forall\,k, \label{cons:sub2matrixlifting_part1}  \\
& &\bm{A}_{mm}=1, \mathrm{for}~m=1,\dots,M+1, \label{cons:sub2matrixlifting_part2}\\
& &\bm{A} \succeq \bm{0}~,~\mathrm{rank}\left(\bm{A}\right)=1. \label{cons:sub2matrixlifting_part3}
\end{eqnarray}
\end{subequations}
Here $\bm{A} \succeq \bm{0}$ indicates that $\bm{A}$ is a positive semidefinite (PSD) matrix. The challenge in solving problem \eqref{prob:sub2matrixlifting} is the nonconvex rank-one constraint. Semidefinite relaxation (SDR) technique is widely adopted to tackle the rank-one constraint in QCQP problems \cite{ma2010semidefinite}. By simply dropping the nonconvex rank-one constraint, SDR relaxes the problem into a convex semidefinite programming (SDP) problem, which can then be solved by CVX. If a feasible $\bm{A}$ with rank one is found, then $\bm{\bar{a}}$ can be obtained by singular value decomposition (SVD) of $\bm{A}$.

However, such a relaxation may not be tight, i.e., the solution obtained by SDR may not satisfy the rank-one constraint. As pointed out in \cite{jiang2019over,yang2018federated}, the performance of SDR degrades sharply as the problem size grows. In our case, when $M$ and/or $K$ is large, the probability of returning a rank-one solution is low. If this is the case, additional steps (i.e., Gaussian randomization) are required to construct a rank-one solution from the higher-rank solution obtained by solving problem \eqref{prob:sub2matrixlifting} \cite{wu2018intelligent, ma2010semidefinite}. However, it is still possible that we fail to find a feasible solution to problem \eqref{prob:DL qcqp} after a large number of Gaussian randomizations.

In other words, dropping the rank-one constraint cannot accurately detect the feasibility of problem \eqref{prob:sub2matrixlifting}. Hence, we propose a novel DC representation for the rank-one constraint, which is capable of guaranteeing the feasibility of the nonconvex rank-one constraint and thus detecting the exact feasibility of problem \eqref{prob:sub2matrixlifting}.
Note that for the PSD matrix $\bm{A}$, the rank-one constraint indicates that
\begin{equation*}
\sigma_{1}(\boldsymbol{A}) > 0~\textrm{and}~\sigma_{i}(\boldsymbol{A})=0, ~~\forall\,i=2, \dots, M+1,
\end{equation*}
where $\sigma_{i}(\boldsymbol{A})$ is the $i$-th largest singular value of $\bm{A}$.
And recall that the trace norm and spectral norm of $\bm{A}$ are defined as 
\begin{equation*}
	\mathrm{Tr}(\bm{A})=\sum_{i=1}^{M+1}\sigma_{i}(\bm{A})~~\textrm{and}~~\|\bm{A}\|=\sigma_1(\bm{M}),
\end{equation*}
respectively. Hence, the rank-one constraint can be equivalently rewritten as the difference of these two convex norms, i.e., 
\begin{eqnarray}
\mathrm{rank}(\bm{A})=1 \Longleftrightarrow \mathrm{Tr}(\bm{A}) -\|\bm{A}\| = 0,~\mathrm{Tr}(\bm{A})>0. \label{eq:rank-one}
\end{eqnarray}
Based on the DC representation for the nonconvex rank-one constraint, problem \eqref{prob:sub2matrixlifting} is reformulated as follows
\begin{eqnarray}\label{prob:sub2dc}
&\underset{\bm{A}\succeq \bm{0}}{\mathrm{minimize}} ~ &g(\bm{A}) := \mathrm{Tr}(\bm{A}) -\|\bm{A}\| \nonumber\\
&\mathrm{subject~to~} &\eqref{cons:sub2matrixlifting_part1}-\eqref{cons:sub2matrixlifting_part2}.
\end{eqnarray}
The rank-one constraint is satisfied when the objective value $g(\bm{A})$ becomes zero. Although problem \eqref{prob:sub2dc} is still nonconvex due to the concave term $-\|\bm{A}\|$ in the objective, we can derive a DC algorithm to solve it in an iterative manner. Specifically, by linearizing the concave term, at iteration $i$ we need to solve the following convex problem
\begin{eqnarray}\label{prob:sub2DCA}
&\underset{\bm{A}\succeq \bm{0}}{\mathrm{minimize}} ~ &\mathrm{Tr}(\bm{A}) - \left\langle\partial\|\bm{A}^{[i-1]}\|,\bm{A}\right\rangle \nonumber\\
&\mathrm{subject~to~} &\eqref{cons:sub2matrixlifting_part1}-\eqref{cons:sub2matrixlifting_part2}.
\end{eqnarray}
where $\bm{A}^{[i-1]}$ is the solution obtained at iteration $i-1$ and $\partial\|\bm{A}^{[i-1]}\|$ denotes the subgradient of spectral norm at point $\bm{A}^{[i-1]}$. Note that one subgradient of $\|\bm{A}\|$ can be efficiently computed as $\bm{q}_1\bm{q}_1^{{\H}}$, where $\bm{q}_1$ is the vector corresponding to the largest singular value $\sigma_1(\bm{A})$ \cite{tao1997convex}. By iteratively solving \eqref{prob:sub2DCA} until the objective function $g(\bm{A})$ in \eqref{prob:sub2dc} becomes zero, we obtain an exact rank-one solution according to \eqref{eq:rank-one}. We design a practical stopping criterion as 
$
\mathrm{Tr}(\bm{A})-\|\bm{A}\| < \epsilon_{{DC}},
$
where $\epsilon_{{DC}}$ is a sufficiently small positive constant.

The convergence characteristic of the iterative DC algorithm for problem \eqref{prob:sub2dc} is presented in the following proposition.
\begin{proposition}\label{proposition:dc convergence}
	The generated sequence $\{g(\bm{A}^{[i]})\}$ is strictly decreasing and the sequence $\{\bm{A}^{[i]}\}$ converges to a critical point of $g$ from an arbitrary initial point $\bm{A}^{[0]}$.
\end{proposition}
\begin{proof}
	Please refer to \cite[Appendix B]{yang2018federated} for more details.
\end{proof}

The proposed DC algorithm can always accurately detect the feasibility of problem \eqref{prob:sub2matrixlifting} by obtaining a feasible $\bm{A}$, which guarantees the optimal objective value of problem \eqref{prob:sub2dc} is zero. This is because the feasible region of problem \eqref{prob:sub2matrixlifting} is always non-empty, at least the obtained solution to problem \eqref{prob:sub2matrixlifting} at iteration $t$ (i.e., $\bm{\Theta}^{(t)}$) is still feasible at iteration $t+1$. Therefore, the strictly decreasing and non-negative sequence $\{g(\bm{A}^{[i]})\}$ can always converge to zero within finite steps.
%Then if the initial point $\bm{\Omega}^{(0)}$ is feasible to problem \eqref{prob:sub2matrixlifting}, \eqref{prob:sub2matrixlifting} is always feasible in subsequent iterations by deduction. The feasibility of \eqref{prob:sub2matrixlifting}

Similarly, we can lift the uplink counterpart \eqref{prob:UL qcqp} as the following feasibility detection problem 
\begin{subequations}\label{prob:DC uplink}
	\begin{eqnarray}
	\hspace{-1em}&\mathrm{find} ~ &\bm{A} \nonumber\\
	%&\mathrm{subject~to~} &\frac{\mathrm{Tr}\left(\bm{R}_{kk}\bm{A}\right) + \left|b_{kk}\right|^2}{\sum_{l \neq k}\alpha_{kl}\left(\mathrm{Tr}\left(\bm{R}_{kl}\bm{A}\right) + \left|b_{kl}\right|^2\right)+c_{k}} \geq \gamma_k,  \nonumber\\
	\hspace{-1em}&\mathrm{subject~to~} &\mathrm{Tr}\left(\bm{R}_{nkk}\bm{A}\right) + \left|b_{nkk}\right|^2 \geq \gamma_k \sum_{l \neq k}^{K}\alpha_{kl}\mathrm{Tr}\left(\bm{R}_{nkl}\bm{A}\right) \nonumber \\
	\hspace{-1em}& &\hspace{1em}+ \gamma_k \left(\sum_{l \neq k}^{K}\alpha_{kl}\left|b_{nkl}\right|^2 + c_{nk}\right),~~\forall\,n, \forall\,k,   \label{cons:part1}\\
	\hspace{-1em}& &\bm{A}_{mm}=1, \mathrm{for}~m=1,\dots,M+1, \label{cons:part2}\\
	\hspace{-1em}& &\bm{A} \succeq \bm{0}~,~\mathrm{rank}\left(\bm{A}\right)=1,
	\end{eqnarray}
\end{subequations}
where $\alpha_{kl} = p_l^{\U}/p_k^{\U}$ and $c_{nk} = \sigma_n^2\left\|\bm{v}_{nk}^{\U}\right\|_2^2/p_k^{\U}$, and iteratively solve the following DC program 
\begin{eqnarray}\label{prob:sub2DCA_uplink}
&\underset{\bm{A}\succeq \bm{0}}{\mathrm{minimize}} ~ &\mathrm{Tr}(\bm{A}) - \left\langle\partial\|\bm{A}^{[i-1]}\|,\bm{A}\right\rangle \nonumber\\
&\mathrm{subject~to~} &\eqref{cons:part1}-\eqref{cons:part2}.
\end{eqnarray}
until the stopping criterion is satisfied to obtain a feasible solution to problem \eqref{prob:DC uplink}.

The overall algorithm for solving problem $\mathscr{P}_\textrm{2}$ (or $\mathscr{P}_\textrm{3}$) is summarized in Algorithm \ref{algo: sub2}. The complexity of Algorithm \ref{algo: sub2} is dominated by iteratively solving the SDP problem \eqref{prob:sub2DCA} (or problem \eqref{prob:sub2DCA_uplink}) and computing the subgradient of $\bm{A}$. In each iteration, problem \eqref{prob:sub2DCA} is solved with complexity $\mathcal{O}(KM^3)$ (or $\mathcal{O}(NKM^3)$ for problem \eqref{prob:sub2DCA_uplink}) using interior-point methods \cite{boyd2004convex}, while the subgradient can be computed by SVD with complexity $\mathcal{O}(M^3)$. Furthermore, it is observed in our simulations that the convergence rate of the iterative procedure is fast (less than $10$ iterations), therefore the overall complexity of Algorithm \ref{algo: sub2} is at most $\mathcal{O}\left(NKM^3\right)$.

\begin{algorithm}[!h]
	\textbf{Input:} $\{\bm{v}_{nk}^{\D}\}, \{\bm{v}_{nk}^{\V}\}, \{p_k^{\U}\}$ and initial point $\bm{A}^{[0]}$ \\
%	1. follow the details defined in \eqref{eq:R, a}-\eqref{eq: alpha, beta} to obtain $\bm{R}_{kl}, b_{kl}, \alpha_{kl}$ and $\beta_{k}$\\
	1. iteratively solve problem \eqref{prob:sub2DCA} (or problem \eqref{prob:sub2DCA_uplink}) until the stopping criterion $\mathrm{Tr}(\bm{A})-\|\bm{A}\| < \epsilon_{{DC}}$ is satisfied \\
	2. decompose $\bm{A}$ as $\bm{A}=\bm{\bar{a}}\bm{\bar{a}}^{\H}$; denote $\bm{\bar{a}} = \left[\bm{a}_0, t_0\right]^{{\sf{T}}}$; then we obtain $\bm{a} = \bm{a}_0/t_0$ and $\bm{\Theta}=\beta\mathrm{diag}\left(\bm{a}^{{\H}}\right)$
	
	\textbf{Output}: $\bm{\Theta}^{\U}$ (or $\bm{\Theta}^{\D}$)
	\caption{DC Algorithm for Feasibility Detection Problem $\mathscr{P}_\textrm{2}$ (or $\mathscr{P}_\textrm{3}$)}
	\label{algo: sub2}
\end{algorithm}

%\begin{figure}
%	\vspace{-2.5em}
%	\includegraphics[width=\linewidth]{fig/three_stage_framework.pdf}
%	\vspace{-9em}
%	\caption{A three-stage framework for $\mathscr{P}_{\textrm{original}}$.}
%	\label{fig: three stage}
%\end{figure}

\subsection{A Three-Stage Framework for Green Edge Inference}\label{subsec:3}
In this subsection, we propose a thorough three-stage framework for the green edge inference problem $\mathscr{P}_{\textrm{original
}}$. In the first stage, we adopt the BSO with mixed $\ell_{1,2}$-norm and DC algorithm to induce the group sparsity structure of uplink/downlink beamforming vectors and optimize the phase-shift matrices.
The obtained solutions serve as a guideline to determine the inference task selection strategy $\A$. In the second stage, an ordering rule is proposed for all tasks according to their priorities, which depend on the structured-sparse beamforming vectors obtained in the first stage as well as several key system parameters  (i.e., channel state information, amplifier efficiency and task computation power). Based on the task ordering, we perform a task selection procedure to finalize $\A$. In the last stage, the beamforming vectors and phase-shift matrices are refined with the finalized task selection strategy $\A$.

%Based on the mixed $\ell_{1,2}$-norm, we can successfully induce the group sparsity of beamforming vectors. However the whole process is not done, because the obtained solutions may contain very small but nonzero coefficients. These small coefficients suggest the execution of corresponding tasks, which may results in energy inefficiency. Therefore in this subsection, we shall propose a thorough three-stage framework for the green edge inference problem $\mathscr{P}$, where the solutions obtained by solving mixed $\ell_{1,2}$-norm based problem serve as a guideline to determine the final task selection strategy $\A$.
%So far we have successfully into tractable convex problems. We are now at the point of proposing the whole algorithm framework for our original green edge inference problem $\sf{P}$. The proposed algorithm is a three stage group sparse beamforming framework. Specifically, in the first stage, we alternatingly optimize  to induce the group sparsity structure of beamforming vectors. 

\textit{Stage 1. Group Sparsity Inducing:} With randomly initialized phase-shift matrices $\bm{\Theta}^{\U}$ and $\bm{\Theta}^{\D}$, we repeatedly solve problems $\mathscr{P}_\textrm{1}$, $\mathscr{P}_\textrm{2}$, and $\mathscr{P}_\textrm{3}$ based on Algorithms \ref{algo: sub1} and \ref{algo: sub2} respectively in an alternating manner until the following stopping criterion is satisfied: the relative improvement of objective values of problem $\mathscr{P}_\textrm{1}$ defined as $\left(P_{\textrm{total}}^{(t)}-P_{\textrm{total}}^{(t+1)}\right)/P_{\textrm{total}}^{(t)} $ is below a predefined threshold $\varepsilon$, where $P_{\textrm{total}}^{(t)}$ and $P_{\textrm{total}}^{(t+1)}$ represent the objective values obtained in iterations $t$ and $t+1$, respectively. The yielded beamforming vectors should have the group sparsity structures. It is worth mentioning that the relative improvement is expected to be non-negative, because $P_{\textrm{total}}$ is non-increasing as proved in Proposition \ref{proposition: convergence}.

\textit{Stage 2. Inference Task Selection:} The next question is how to determine the task selection strategy $\A$. In fact, it is inappropriate to directly use the beamforming vectors obtained in \textit{Stage 1} to finalize $\A$, as the vectors may contain some very small but nonzero coefficients. As illustrated in \eqref{eq: power}, these nonzero coefficients indicate that the corresponding tasks are performed by the BSs, which may result in unnecessary computation power consumption. To address this issue, we utilize the obtained group-structured beamforming vectors as well as other prior information to provide a guideline to determine set $\A$.

For ease of exposition, we define the set of all task indices as $\{(n,k)|n\in\N,~k\in\K\}$. The task indexed by $(n,k)$ is considered to be \textit{active} if $k\in\A_n$, and \textit{inactive} otherwise. Specifically, the priority of task $(n,k)$ is defined as
\begin{eqnarray}
\tau_{nk}= \sqrt{\frac{\left\|\left[\bm{g}_{nk}^{\U}, \bm{g}_{nk}^{\D}\right]\right\|_2^2\eta_n}{{P}_{nk}^{\sf{c}}}} \left\|\bm{v}_{nk} \right\|_2. \label{eq:order rule}
\end{eqnarray}
We sort all $NK$ tasks in a descending order according to their priorities, i.e., $\tau_{\pi_1}\geq\tau_{\pi_2}\dots\geq\tau_{\pi_{NK}}$, where $\pi$ is a permutation of task indices $(n,k)$'s. Intuitively, if BS $n$ has a higher power amplifier efficiency, a higher channel gain, and a higher beamforming gain with respect to MD $k$, but lower computation power consumption, task $(n,k)$ has a higher priority. A higher $\tau_{nk}$ implies that task $(n,k)$ is more power-efficient, therefore it is more likely to be activated.

To finalize the task selection strategy, we need to detect the feasibility of a sequence of problems 
\begin{eqnarray}\label{prob:revision}
&\mathrm{find}~~&\{\bm{v}_{nk}^{\D}\}, \{\bm{v}_{nk}^{\U}\}, \{p_k^{\U}\} \nonumber \\ 
&\mathrm{subject~to} ~~&\eqref{cons:sinr_downlink}-\eqref{cons:power_uplink}, \nonumber\\
%&\quad \quad \quad \quad \quad	&\textrm{SINR}_k^{\D} \geq \gamma_k^{\D}, ~~\forall k, \nonumber\\
%& \quad \quad \quad \quad \quad  &p_k^{\U} \leq P_{k,\textrm{max}}^{\U}, ~~\forall k, \nonumber\\
%& \quad \quad \quad \quad \quad &\sum_{k=1}^{K} \left\|\bm{v}_{nk}^{\D} \right\|_2^2 \leq P_{n,\textrm{max}}^{\D}, ~~\forall n, \nonumber \\
& \quad \quad \quad \quad \quad &\bm{v}_{\pi^{[j]}} = \bm{0},
\end{eqnarray}
where $\pi^{[j]} = \{\pi_{j+1},\dots,\pi_{NK}\}$ denotes the inactive task set at iteration $j$, and $\bm{v}_{\pi^{[j]}} = \bm{0}$ represents that all coefficients in those beamforming groups $\bm{v}_{nk}$'s with index $(n,k)\in {\pi^{[j]}}$ are set to zero. Note that the number of active tasks is within $[K,NK]$. Starting with $j=K$, we terminate the feasibility detection procedure and return $\pi^{[j]}$ until problem $\eqref{prob:revision}$ is feasible. The task selection strategy $\A$ can be easily obtained from $\pi^{[j]}$.
 
Comparing problem \eqref{prob:revision} to $\mathscr{P}_1$, as the added constraint $\bm{v}_{\pi^{[j]}} = \bm{0}$ is convex, we can solve problem \eqref{prob:revision} using Algorithm \ref{algo: sub1}. Details are thus omitted here for brevity.

\textit{Stage 3. Optimization Variables Refinement:} After determining the task selection strategy $\A$, the computation power $\sum_{n\in\N}\sum_{k\in\A_n}P_{nk}^{\sf{c}}$ is a constant and can be removed from the objective function. The uplink receive beamforming vectors and downlink transmit beamforming vectors are determined by coordinated beamforming among the BSs, and the phase-shift matrices at the RIS need to be refined as well. We solve the following problem to refine these optimization variables
\begin{eqnarray}\label{prob:final refine variables}
&\underset{\{\bm{v}_{nk}^{\D}\}, \{\bm{v}_{nk}^{\U}\}, \{p_k^{\U}\}, \bm{\Theta}^{\D}, \bm{\Theta}^{\U}}{\mathrm{minimize}} ~~&\sum_{k=1}^{K}p_k^{\U} +  
\sum_{n=1}^{N} \sum_{k=1}^{K}\frac{1}{\eta_n} \|\bm{v}_{nk}^{\D}\|_2^2\nonumber\\
&\mathrm{subject~to} ~~&\eqref{cons:sinr_downlink}-\eqref{cons:phase_uplink},\nonumber\\
%&\quad \quad \quad \quad \quad  &\textrm{SINR}_k^{\U} \geq \gamma_k^{\U}, ~~\forall k, \nonumber\\
%& \quad \quad \quad \quad \quad  &\sum_{k=1}^{K} \|\bm{v}_{nk}^{\D} \|_2^2 \leq P_{n,\textrm{max}}^{\D}, ~~\forall n, \nonumber\\
%& \quad \quad \quad \quad \quad &p_k^{\U} \leq P_{k,\textrm{max}}^{\U}, ~~\forall k, \nonumber\\
%& \quad \quad \quad \quad \quad &0\leq \theta_m^{\U}\leq 2\pi, ~~\forall m, \nonumber \\
%& \quad \quad \quad \quad \quad &0\leq \theta_m^{\D}\leq 2\pi, ~~\forall m, \nonumber \\
& \quad \quad \quad \quad \quad &\bm{v}_{\pi^{[j]}} = \bm{0}.
\end{eqnarray}
Comparing problem \eqref{prob:final refine variables} to $\mathscr{P}$, since the added constraint $\bm{v}_{\pi^{[j]}} = \bm{0}$ is convex, the BSO with mixed $\ell_{1,2}$-norm and DC algorithm to solve problem $\mathscr{P}$ is also applicable here to obtain the solutions. Details are thus omitted here.
%In particular, the alternating process terminates when one of the following stopping criteria is satisfied: 
%\begin{itemize}
%	\item The improvement of objective values in consecutive iterations is below a threshold $\epsilon$.
%	\item Problem \eqref{prob: optimize theta downlink} or \eqref{prob: optimize theta uplink} becomes infeasible.
%\end{itemize}

The overall algorithm for green edge inference is summarized in Algorithm \ref{algo:three-stage alternating}. By denoting the required iterations for the BSO to converge as $T$, the computational complexity involved in \textit{Stage 1} is $\mathcal{O}\left(T(L^{3.5}K^{3.5}+NKM^3)\right)$, where $\mathcal{O}(L^{3.5}K^{3.5}+ NKM^3)$ is the complexity at each iteration discussed in Sections \ref{subsec:1} and \ref{subsec:2}. We need to solve a sequence of SOCP problems in \textit{Stage 2}, and hence the worst-case complexity is $\mathcal{O}\left(NK(L^{3.5}K^{3.5})\right)$. The complexity involved in \textit{Stage 3} is the same as that in \textit{Stage 1}. Therefore, the overall complexity for the proposed three-stage framework is $\mathcal{O}\left(T(L^{3.5}K^{3.5}+NKM^3)+NL^{3.5}K^{4.5}\right)$.
\begin{algorithm}[!h]
	\textbf{Input:} initial phase-shift matrices $\bm{\Theta}^{\U}$ and $\bm{\Theta}^{\D}$, and threshold $\varepsilon$ \\
	\textit{Stage 1:} Alternatively optimize $\B_1$, $\B_2$, and $\B_3$\\
	\While{the improvement of the objective in problem $\mathscr{P}_\textrm{1}$ is greater than $\varepsilon$}{1. solve $\mathscr{P}_\textrm{1}$ for $\bm{v}_{nk}^{\U}, p_k^{\U}, \bm{v}_{nk}^{\D}$ using Algorithm \ref{algo: sub1}\\ 2. solve $\mathscr{P}_\textrm{2}$ for $\bm{\Theta}^{\D}$ using Algorithm \ref{algo: sub2} \\
	3. solve $\mathscr{P}_\textrm{3}$ for $\bm{\Theta}^{\U}$ using Algorithm \ref{algo: sub2}} 
	\textit{Stage 2:} Determine the inference task selection \\1. compute task priorities based on \eqref{eq:order rule} and sort all tasks in a descending order according to their priorities\\ 2. iteratively solve problem \eqref{prob:revision} until feasible to finalize the task selection strategy $\A$\\
	\textit{Stage 3:} Solve problem \eqref{prob:final refine variables} to refine variables $\{\bm{v}_{nk}^{\D}\}, \{\bm{v}_{nk}^{\U}\}, \{p_k^{\U}\}, \bm{\Theta}^{\D}, \bm{\Theta}^{\U}$  \\
	\textbf{Output}: $\A$, $\{\bm{v}_{nk}^{\D}\}, \{\bm{v}_{nk}^{\U}\}, \{p_k^{\U}\}, \bm{\Theta}^{\D}, \bm{\Theta}^{\U}$
	\caption{BSO with Mixed $\ell_{1,2}$-Norm and DC Based Three-Stage Framework for Nonconvex Combinatorial Problem $\mathscr{P}_{\textrm{original}}$}
	\label{algo:three-stage alternating}
\end{algorithm}

\section{Simulation Results}\label{sec: simulation}
In this section, we present the simulation results to verify the effectiveness of our proposed algorithm. We consider a network with five $8$-antenna BSs and six MDs uniformly and independently distributed in a square region of 300$\,$m$\times$300\,m. An RIS with $30$ reflecting elements is located at the 3-dimensional coordinate $(150,0,20)$. In addition, the BSs are with height $30\,$m (i.e., the coordinates of the BSs are $(x,y,30)$), while the MDs are with height $0\,$m (i.e., the coordinates of the MDs are $(x,y,0)$).

We consider the following distance-dependent path loss model
$
L(d)=T_{0}\left(\frac{d}{d_{0}}\right)^{-\alpha},
$
where $T_{0}$ is the constant path loss at the reference distance $d_{0}$, $d$ is the Euclidean distance between the transceivers, and $\alpha$ is the path loss exponent. Each antenna of the BSs is assumed to have an isotropic radiation pattern (i.e., 0\,dBi antenna gain), while each element of the RIS has a 3\,dBi gain because it reflects signals only in its front half-space \cite{wu2019joint, griffin2009complete}. Moreover, Rayleigh fading is considered as the small-scale fading model for all channels. Specifically, the channel coefficients are given by
$\bm{h}_{\textrm{d},nk}^{{x}}= \sqrt{L(d_{\textrm{BU}})}\bm{\xi}_1$, $ \bm{G}_n^{x}=\sqrt{10^{0.3}L(d_{\textrm{RB}})}\bm{\Gamma}$, and $\bm{h}_{\textrm{r},k}^{{x}}=\sqrt{10^{0.3}L(d_{\textrm{RU}})}\bm{\xi}_2$, where $\bm{\xi}_1, \bm{\xi}_2\sim\mathcal{CN}(\bm{0},\bm{I})$, $\bm{\Gamma}\sim\mathcal{CN}(\bm{0},\bm{I})$, and superscript ${x}\in\{\U,\D\}$.
%To account for the heterogeneity of the MDs' tasks and BSs' computational capabilities, $P_{nk}^{\sf{c}}$ is uniformly sampled from $0.45\,$W to $0.60\,$W. 
Without specified otherwise, we set $T_{0}=-10\,$dB, $d_0=1\,$m,  $P_{n,\textrm{max}}^{\D}=1\,$W, $P_{n,\mathrm{max}}^{\U}=1\,$W, $P_{nk}^{\sf{c}}=0.45\,$W, $\sigma_k^2=-53\,$dBm, $\sigma_n^2=-63\,$dBm, $\eta_n=0.25$, $\beta=1$, $\epsilon_{{DC}}=10^{-6}$,  $\varepsilon=10^{-2}$, and $\gamma_k^{\U}=\frac{1}{2}\gamma_k^{\D}$. The path loss exponent $\alpha$ is set as $2,3.5$, and $2.5$ for BS-RIS channel, BS-MD channel, and RIS-MD channel, respectively.
%Each point in Figs. \ref{fig:exp1}~-~\ref{fig: 2}  is averaged over 200 independently generated channel realizations unless otherwise specified.

We compare the proposed BSO with mixed $\ell_{1,2}$-norm and DC algorithm (abbreviated as BSO-$\ell_{1,2}$-DC) with the following benchmarks.
\begin{itemize}
	\item \textbf{Without-RIS}: Without the deployment of an RIS, the equivalent channels in \eqref{eq: uplink channel} and \eqref{eq: downlink channel} contain only the direct link, i.e., $\bm{h}_{\textrm{r},k}^{\U}=\bm{h}_{\textrm{r},k}^{\D}=\bm{0}, \forall\,k$. As we do not need to optimize phase shifts in this case, the alternating process in \textit{Stage 1} is simplified to solve $\mathscr{P}_{1}$ only once.
	\item \textbf{BSO with mixed $\ell_{1,2}$-norm and Random Phase} (abbreviated as BSO-$\ell_{1,2}$-RP): In this case, the phase shifts of all reflecting elements in both uplink and downlink transmissions are randomly chosen from $[0,2\pi)$ and then used to solve problem $\mathscr{P}_{\textrm{1}}$. We do not solve problems $\mathscr{P}_{\textrm{2}}$ and $\mathscr{P}_{\textrm{3}}$ in \textit{Stage 1} subsequently to optimize the phase shifts. This benchmark is designed to reveal the necessity of optimizing the phase-shift matrices.
	\item \textbf{BSO with mixed $\ell_{1,2}$-norm and SDR} (abbreviated as BSO-$\ell_{1,2}$-SDR): In this case, the nonconvex rank-one constraints in \eqref{prob:sub2matrixlifting} and \eqref{prob:DC uplink} are dropped. Gaussian randomization is then adopted to obtain a feasible solution to problems $\mathscr{P}_\textrm{2}$ and $\mathscr{P}_\textrm{3}$. The number of randomly generated vectors for Gaussian randomization is set as $500$. If Gaussian randomization fails to find a feasible solution, we terminate the alternating process in \textit{Stage 1}.
\end{itemize}

\subsection{Effectiveness of Deploying An RIS}
In this subsection, we compare our RIS-aided communication system with the conventional one without the assistance of an RIS. For fair comparison, we do not explicitly optimize the phase-shift matrices, i.e., we compare the performance of Without-RIS and BSO-$\ell_{1,2}$-RP.

We first study the relationship between the feasible probability of problem $\mathscr{P}_{\textrm{original}}$ and the target SINR $\gamma_{k}^{\D}$. The feasible probability of problem $\mathscr{P}_{\textrm{original}}$ is defined as 
\begin{eqnarray}
\mathbb{P}\{\mathscr{P}_{\textrm{original}}\,\textrm{is\,feasible}\}\!=\!\frac{\textrm{number of cases~} \mathscr{P}_{\textrm{original}}~\textrm{is~feasible}}{\textrm{number of total test cases}}. \nonumber
\end{eqnarray}
As the target SINR requirements become more stringent, i.e., larger values of $\gamma_k^{\U}$ and $\gamma_k^{\D}$, the feasibility probability of problem $\mathscr{P}_{\textrm{original}}$ is expected to decline. Results illustrated in Table \ref{table: feasible probability} are averaged over 200 independently generated channel realizations.
We can see that the conventional system without RIS fails to support those settings with a target SINR being higher than 0\,dB, while the RIS-aided system can still support with a high probability. In terms of the maximum SINR that the communication systems can support, we observe that there exists at least a 15\,dB gain of the RIS-aided system over the system without RIS.

Table \ref{table: power consumption} illustrates the overall network power consumption of systems with and without RIS. Under the same SINR requirement, it is observed that BSO-$\ell_{1,2}$-RP yields a significantly lower power consumption. The supreme performance gain demonstrates that the deployment of an RIS in wireless communication systems can greatly boost the SINR and in turn reduce the overall network power consumption.

\begin{table}[!t]
	\centering
	\caption{Feasible Probability versus Target SINR}
	\begin{tabular}{l|c|c|c|c|c|c|c}
		\hline
		Target SINR [dB] & -20 & -15 & -10 & -5 & 0 & 5 & 10\\ \hline
		Without-RIS & 1.00 & 0.84 & 0.48 & 0.08 & 0 & 0 & 0\\ \hline
		BSO-$\ell_{1,2}$-RP & \textbf{1.00} & \textbf{1.00} & \textbf{1.00} & \textbf{1.00} & \textbf{0.97} & \textbf{0.84} & \textbf{0.42}\\ \hline
		
	\end{tabular}
	\label{table: feasible probability}
\end{table}

\begin{table}[!t]
	\centering
	\caption{Overall Power Consumption versus Target SINR}
	\begin{tabular}{l|c|c|c|c|c|c|c}
		\hline
		Target SINR [dB] & -20 & -15 & -10 & -5 & 0 & 5 & 10\\ \hline
		Without-RIS [W] & 3.04 & 3.52 & 5.17 & 9.60 & N/A & N/A & N/A \\ \hline
		BSO-$\ell_{1,2}$-RP [W] & \textbf{3.02} & \textbf{3.12} & \textbf{3.32} & \textbf{3.75} & \textbf{4.72} & \textbf{7.66} & \textbf{14.79}\\ \hline
		
	\end{tabular}
	\label{table: power consumption}
\end{table}

\begin{figure}[!t]
	\centering
	\includegraphics[width=8.5cm,height=6.5cm]{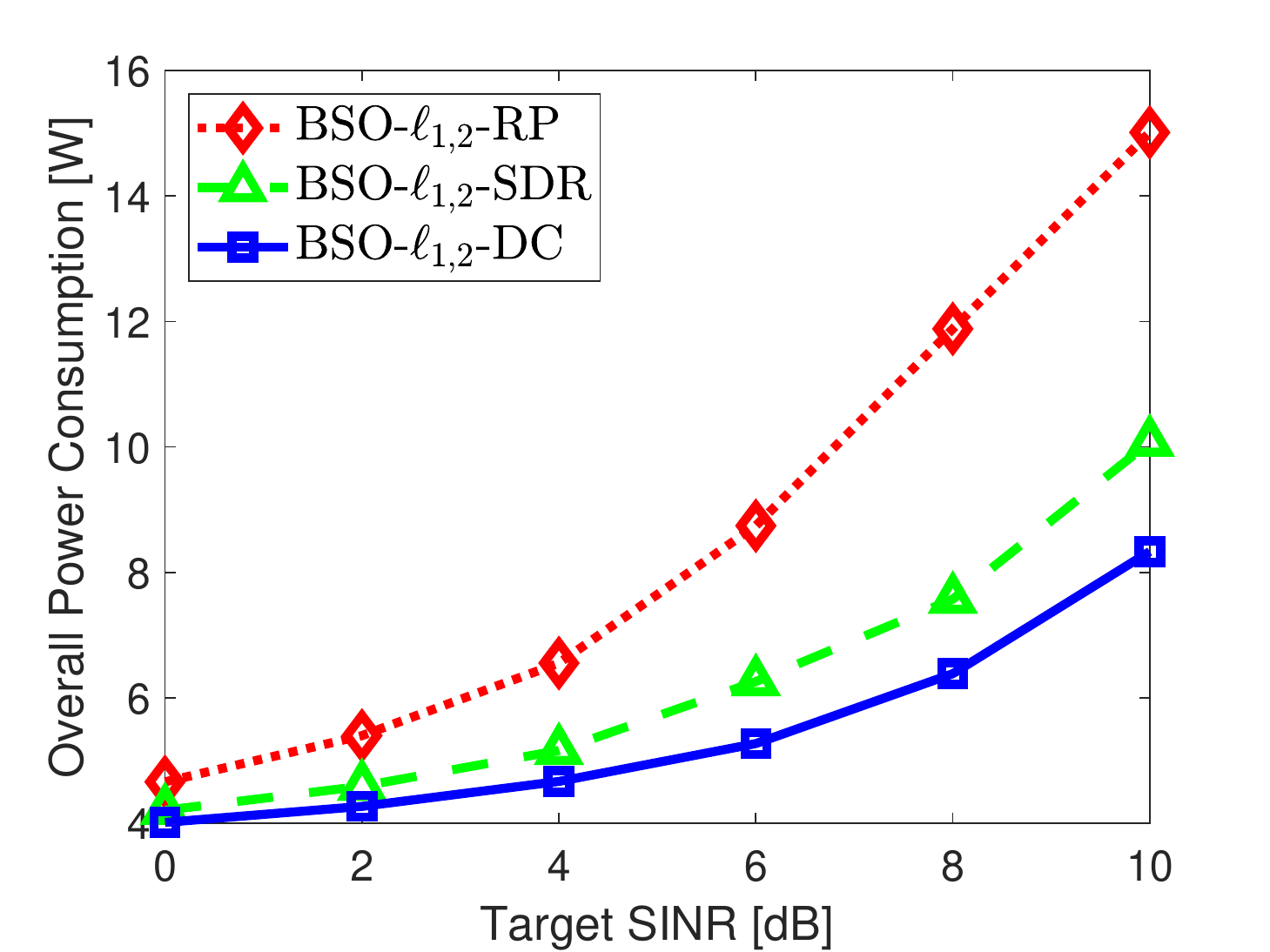}
	\caption{The overall network power consumption versus target SINR $\gamma_{k}^{\D}$.}
	\label{fig: 1}
\end{figure}

\begin{figure}
	\centering
	\includegraphics[width=8.5cm,height=6.5cm]{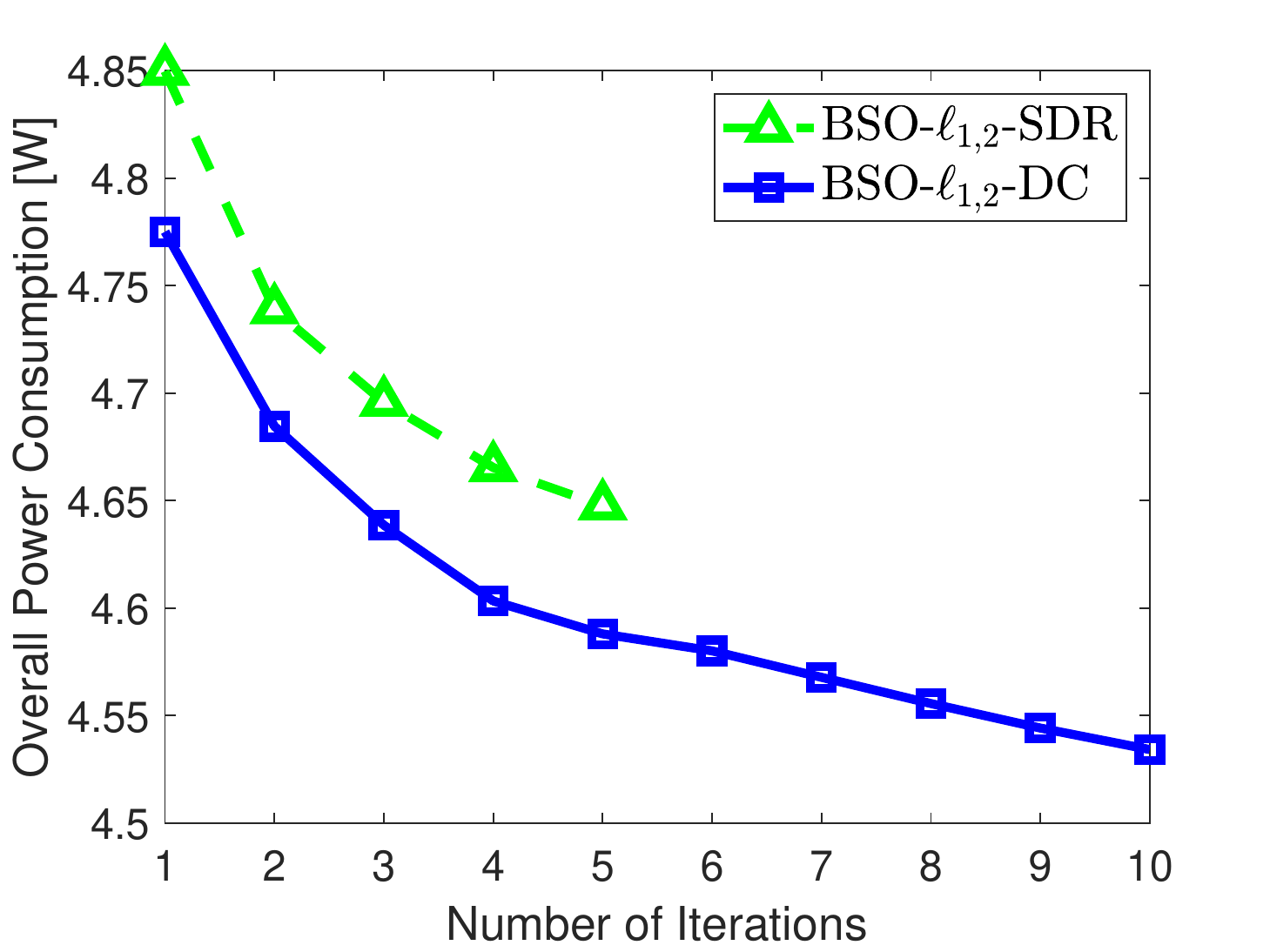}
	\caption{Convergence behaviors of both BSO-$\ell_{1,2}$-DC and BSO-$\ell_{1,2}$-SDR algorithms.}
	\label{fig: convergence}
\end{figure}

\begin{figure}[!t]
	\centering
	\begin{subfigure}{\linewidth}
		\centering
		\includegraphics[width=8.5cm,height=6.5cm]{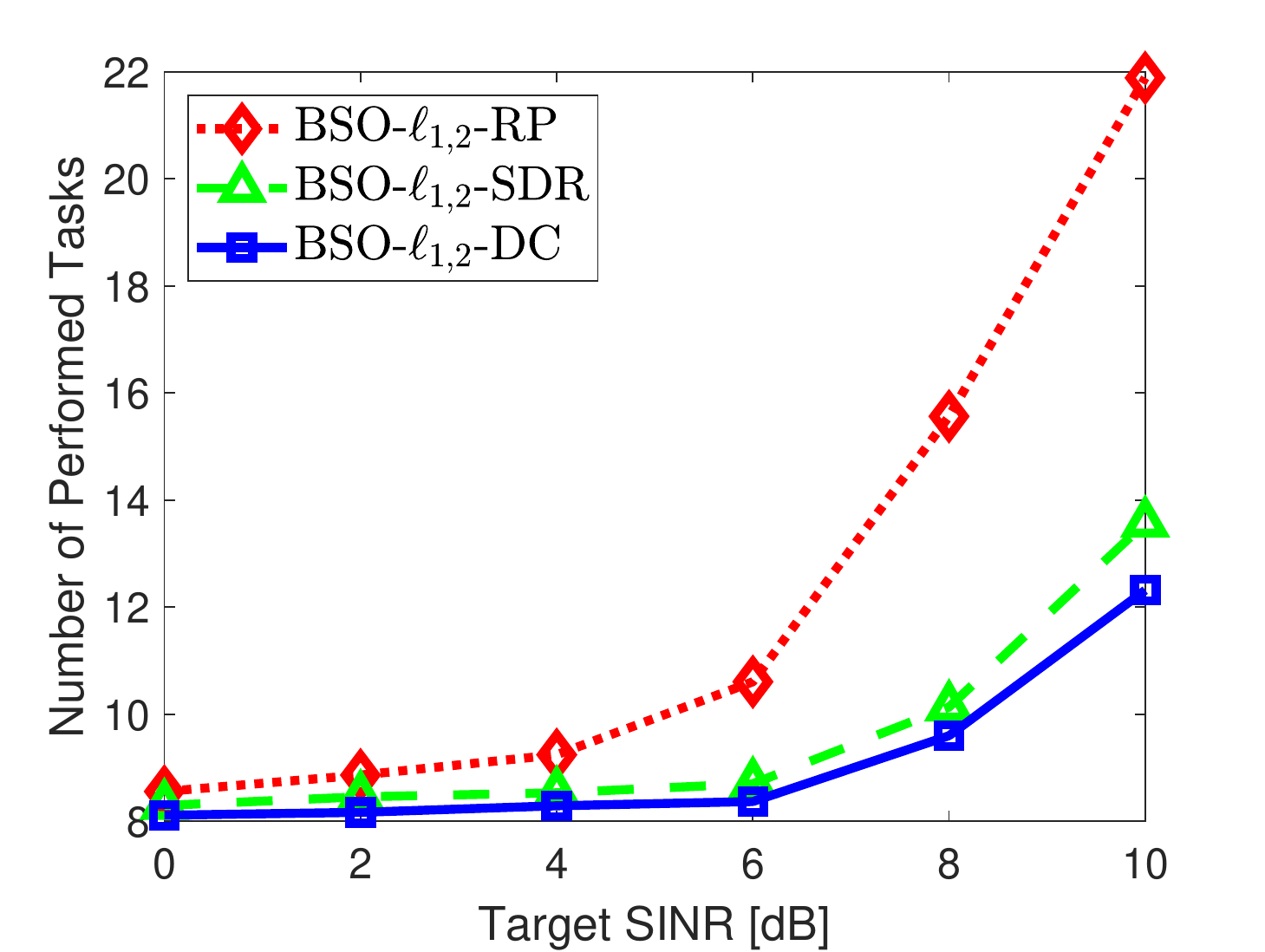}
		\caption{Number of performed tasks versus target SINR.}
		\label{fig:num tasks}
	\end{subfigure}
	\begin{subfigure}{\linewidth}
	\centering
	\includegraphics[width=8.5cm,height=6.5cm]{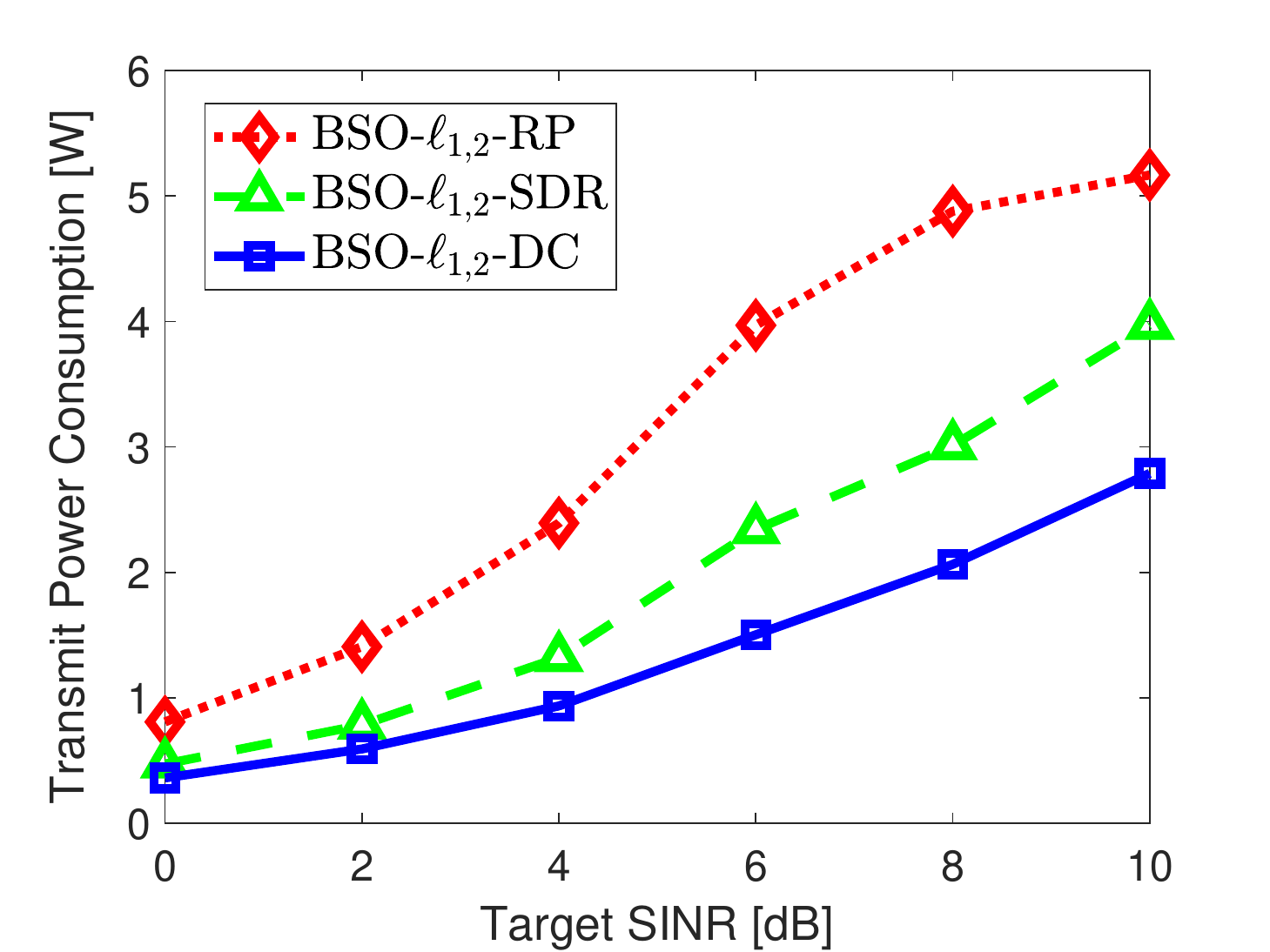}
	\caption{Transmit power consumption versus target SINR.}
	\label{fig:trans power}
	\end{subfigure}
	\caption{Different components of the overall network power consumption versus target SINR $\gamma_{k}^{\D}$.}
	\label{fig: 2}
\end{figure}

\subsection{Superiority of The Proposed BSO-$\ell_{1,2}$-DC Algorithm}
In this subsection, we show the superiority of our proposed BSO-$\ell_{1,2}$-DC algorithm. The overall network power consumption shown in Fig. \ref{fig: 1} is averaged over 150 independent channel realizations.

The first observation is that both the BSO-$\ell_{1,2}$-DC and BSO-$\ell_{1,2}$-SDR algorithms significantly outperform BSO-$\ell_{1,2}$-RP, which demonstrates that dynamically optimizing the phase-shift matrices according to the beamforming vectors can reduce the network power consumption to a large extent. In addition, we observe that the proposed BSO-$\ell_{1,2}$-DC algorithm yields a much lower overall network power consumption and is more energy efficient than the BSO-$\ell_{1,2}$-SDR algorithm. Given an overall power budget (e.g., 6\,W), BSO-$\ell_{1,2}$-DC can achieve around 2\,dB higher SINR for the MDs than BSO-$\ell_{1,2}$-SDR. Such a performance gain is mainly because BSO-$\ell_{1,2}$-SDR may early terminate the alternating BSO process in \textit{Stage 1} and cannot further proceed to find feasible solutions to problems $\mathscr{P}_\textrm{2}$ and/or $\mathscr{P}_\textrm{3}$. 
%However, as we have discussed in Section , DC can always obtain feasible solutions to problems $\mathscr{P}_\textrm{2}$ and $\mathscr{P}_\textrm{3}$.
To make this more explicit, Fig. \ref{fig: convergence} shows the objective values of problem $\mathscr{P}_1$ in the first $10$ alternating iterations in a specific channel realization. It is observed that as the BSO approach proceeds, the overall network power consumption for both BSO-$\ell_{1,2}$-DC and BSO-$\ell_{1,2}$-SDR algorithms are non-increasing, which validates our analysis in Proposition \ref{proposition: convergence}. It is also observed that the BSO-$\ell_{1,2}$-SDR algorithm terminates at the 5th iteration, as SDR fails to obtain feasible solutions to problems $\mathscr{P}_{\textrm{2}}$ and/or $\mathscr{P}_{\textrm{3}}$ even with Gaussian randomization techniques. In contrast, DC can always yield feasible solutions as we have discussed in Section \ref{subsec:2}, and therefore BSO-$\ell_{1,2}$-DC terminates the alternating process only when the consecutive iterations make little progress.

Another interesting point worth mentioning in Fig. \ref{fig: 1} is that the performance gaps between BSO-$\ell_{1,2}$-DC and other algorithms are getting larger as the value of the target SINR increases, which indicates that BSO-$\ell_{1,2}$-DC is especially appealing when high-quality services are required by the MDs. This is because a higher target SINR leads to a narrower feasible region of problems $\mathscr{P}_\textrm{2}$ and $\mathscr{P}_\textrm{3}$, making SDR less likely to find a feasible solution. In short, Fig. \ref{fig: 1} shows that the proposed BSO-$\ell_{1,2}$-DC is able to reduce the overall network power consumption by 20\% in the low SINR regime, and by up to 45\% in the high SINR regime.

The number of tasks performed by all the BSs and the transmit power consumption versus the target SINR are shown in Fig. \ref{fig:num tasks} and Fig. \ref{fig:trans power}, respectively. As the target SINR increases, both the number of performed tasks and transmit power consumption increase. It is observed in Fig. \ref{fig:num tasks} that BSO-$\ell_{1,2}$-DC can always perform fewer tasks to satisfy a certain target SINR. In other words, the long-lasting alternating iterations of BSO-$\ell_{1,2}$-DC shown in Fig. \ref{fig: convergence} helps promote the group sparsity structure of beamforming vectors, thereby achieving lower computation power consumption. In terms of the transmit power consumption depicted in Fig. \ref{fig:trans power}, we make the similar observation that BSO-$\ell_{1,2}$-DC yields the lowest transmit power consumption. Finally, it is also observed that the performance gaps between BSO-$\ell_{1,2}$-DC and other algorithms tend to be larger in the high SINR regime.

%It is worth mentioning that in the low QoS regime (i.e., -2dB), the average number of processed tasks is $5$ for DC and SDR algorithms, which is exactly the number of MDs. This implies that both DC and SDR achieve nearly optimal performance, i.e., tasks from MDs are performed only once. 

%\begin{figure}[!t]
%	\includegraphics[width=\linewidth]{}
%	\caption{The Impact of $\lambda$}
%	\label{fig:exp3}
%\end{figure}

%\subsection{The Impact of Factor \texorpdfstring{$\lambda$}{text}}
%In this experiment, we explore how the power tradeoff factor $\lambda$ will affect the performance. MDs' target QoS is fixed at 0dB, for which we have already check the feasibility of the original problem $\sf{P}$ in our first experiment. The interested factor $\lambda$ ranges from 1 to 21. As have been discussed in Section, a larger $\lambda$ indicates that we emphasize more on MDs' transmit power consumption, thereby reducing MDs' power will be of more importance. 
%
%MDs' and BSs' power consumption with different $\lambda$ are demonstrated in Table . It is observed that for three algorithms, the total power consumption of MDs decreases while that of increases , as $\lambda$ grows. This is in accordance with our expectations. The total power consumption results versus the value of $\lambda$ are plotted in Fig. .

\section{Conclusions}\label{sec: conclusion}
In this paper, we investigated an RIS-aided edge inference system with multiple BSs cooperatively serving multiple MDs, taking into account both uplink and downlink transmissions. The design of an energy-efficient edge inference system was formulated as a joint uplink and downlink beamforming, transmit power and phase-shift matrices design problem. A BSO approach was proposed to decouple the optimization variables. For efficient algorithm design, mixed $\ell_{1,2}$-norm was adopted to induce group sparsity of uplink/downlink beamforming vectors, while the matrix lifting and DC techniques were exploited to handle the nonconvex rank-one constraint and in turn solve the phase-shift matrix optimization problems. Through numerical simulations, we demonstrated that the deployment of an RIS significantly reduces the overall network power consumption. Furthermore, the effectiveness of DC algorithm for inducing low-rank solutions was verified. We also clarified the convergence behavior of the proposed BSO approach in this paper. For future work, we will study the precise convergence rate of the proposed algorithm.

\bibliographystyle{IEEEtran}
\bibliography{ref.bib}

\nolinenumbers
\end{document}